\documentclass[11pt]{article}
\usepackage{fullpage}
\usepackage{epsfig}
\usepackage{amsfonts}
\usepackage{amssymb}
\usepackage{amstext}
\usepackage{amsmath}
\usepackage{xspace}
\usepackage{theorem}
\usepackage{graphicx}
\usepackage{graphics}
\usepackage{colordvi}
\usepackage{colordvi}

\newcommand{\myparskip}{3pt}
\parskip \myparskip
\newcommand{\oset}{{\mathcal{O}}}

\newcommand{\inedges}[1]{\Delta^-_{\oset}(#1)}
\newcommand{\outedges}[1]{\Delta^+_{\oset}(#1)}
\newcommand{\notx}{\overline x}
\def\Nj{\mathcal N_j}

\newcommand{\MMA}{{\sc{Max-Min Allocation}}\xspace}

\newenvironment{proofof}[1]{\noindent{\bf Proof of #1.}}%
        {\hspace*{\fill}$\Box$\par\vspace{4mm}}

\newenvironment{properties}[2][0]
{\renewcommand{\theenumi}{#2\arabic{enumi}}
\begin{enumerate} \setcounter{enumi}{#1}}{\end{enumerate}\renewcommand{\theenumi}{\arabic{enumi}}}

\newcommand{\bor}{\vee}

\newcommand{\opt}{\mbox{\sf OPT}}

\newcommand{\set}[1]{\left\{ #1 \right\}}
\newcommand{\sse}{\subseteq}

\newcommand{\tset}{{\mathcal M}}
\newcommand{\iset}{{\mathcal{I}}}
\newcommand{\pset}{{\mathcal{P}}}

\newcommand{\lset}{{\mathcal{L}}}

\newcommand{\aset}{{\mathcal{A}}}
\newcommand{\cset}{{\mathcal{C}}}

\newcommand{\rset}{{\mathcal{R}}}

\newcommand{\I}{{\mathcal I}}

\newcommand{\sset}{{\mathcal{S}}}

\newcommand{\be}{\begin{enumerate}}
\newcommand{\ee}{\end{enumerate}}
\newcommand{\bd}{\begin{description}}
\newcommand{\ed}{\end{description}}
\newcommand{\bi}{\begin{itemize}}
\newcommand{\ei}{\end{itemize}}

\newtheorem{lemma}{Lemma}
\newtheorem{theorem}{Theorem}

\newtheorem{corollary}{Corollary}
\newtheorem{claim}{Claim}

\newenvironment{definition}{{\bf Definition}: }{}
\newenvironment{proof}{\par \smallskip{\bf Proof:}}{\hfill\stopproof}
\def\stopproof{\square}
\def\square{\vbox{\hrule height.2pt\hbox{\vrule width.2pt height5pt \kern5pt
\vrule width.2pt} \hrule height.2pt}}

\renewcommand{\phi}{\varphi}
\newcommand{\eps}{\epsilon}

\newcommand{\half}{\ensuremath{\frac{1}{2}}}

\newcommand{\poly}{\operatorname{poly}}

\newcommand{\F}{\ensuremath{\mathbb F}}

\def\A{\mathcal A}
\def\eps{\epsilon}

\title{On Allocating Goods to Maximize Fairness}
\author{
Deeparnab Chakrabarty\footnote{Department of Combinatorics and Optimization, University of Waterloo, Waterloo, Ontario, Canada. Email: {\tt deepc@math.uwaterloo.ca}.}\and
Julia Chuzhoy\footnote{Toyota Technological Institute, Chicago, IL
60637. Email: {\tt cjulia@tti-c.org}.} \and Sanjeev Khanna\footnote{Dept. of Computer \&
Information Science, University of Pennsylvania, Philadelphia, PA
19104. Email: {\tt sanjeev@cis.upenn.edu}. Supported in part by a Guggenheim Fellowship,
an IBM Faculty Award, and by NSF Award CCF-0635084. }}

\begin{document}
\maketitle

\begin{abstract}
Given a set $\aset$ of $m$ agents and a set $I$ of $n$ items, where agent $A\in \aset$ has utility $u_{A,i}$ for item $i\in I$, our goal is to allocate items to agents to maximize fairness. Specifically, the utility of an agent is the sum of the utilities for items it receives, and we seek to maximize the minimum utility of any agent. While this problem has received much attention recently, its approximability has not been well-understood thus far: the best known approximation algorithm achieves an
$\tilde{O}(\sqrt{m})$-approximation, and in contrast, the best known hardness of approximation stands at $2$.

Our main result is an approximation algorithm that achieves an $\tilde{O}(n^{\eps})$ approximation for any $\eps=\Omega(\log\log n/\log n)$ in time $n^{O(1/\eps)}$. In particular, we obtain poly-logarithmic approximation in quasi-polynomial time, and for every constant $\eps > 0$, we obtain $\tilde O(n^{\eps})$-approximation in polynomial time. An interesting technical aspect of our algorithm is that we use as a building block a linear program whose integrality gap is $\Omega(\sqrt m)$. We bypass this obstacle by iteratively using the solutions produced by the LP to construct new instances with significantly smaller integrality gaps, eventually obtaining the desired approximation.

We also investigate the special case of the problem, where every item has a non-zero utility for at most two agents.
We show that even in this restricted setting the problem is hard to approximate upto any factor better than $2$, and show a factor $(2+\eps)$-approximation algorithm running in time
$\mbox{poly}(n,1/\eps)$ for any $\eps>0$. This special case can be cast as a graph edge orientation problem, and our algorithm can be viewed as a generalization of Eulerian orientations to weighted graphs.

\end{abstract}

\section{Introduction}

In this paper we consider the problem of allocating indivisible goods to a set of agents with the objective to maximize the minimum utility among all agents.
In particular, we are given a set $\A$ of $m$ agents, a set $I$ of $n$ indivisible items, and non-negative utilities $u_{A,i}$ for each agent $A$ and item $i$. The total utility of an agent $A$ for a subset $S\subseteq I$ of items is $u_A(S) := \sum_{i\in S} u_{A,i}$, that is,  the utility function is additive. An allocation of items is a function $\Phi: \A \to 2^I$ such that an item is allocated to at most one agent, that is,
$\Phi(A) \cap \Phi(A') = \varnothing$ whenever $A \neq A'$. The \MMA  problem is to find an allocation $\Phi$  of items which maximizes
$\min_{A\in \A}\set{ u_A(\Phi(A))}$.

The \MMA problem arises naturally in the context of {\em fair} allocation of indivisible resources where maximizing the utility of the least `happy' person is arguably a suitable notion of fairness. This is in contrast to allocation problems where the goal is to maximize the {\em total} utility of agents, a problem that admits a trivial solution for any additive utility function: allocate each item to a person that derives the highest utility from it. The \MMA problem may be viewed as a `dual' problem to the well studied min-max (also known as makespan) scheduling problem where the goal is to find an allocation minimizing the utility (load) of the maximum utility agent (machine).

The \MMA problem was indeed first studied as a machine scheduling problem where the minimum completion time has to be maximized. Woeginger \cite{Woe1} and Epstein and Sgall \cite{ES} gave polynomial time approximation schemes (PTAS) for the cases when all agents (machines) have identical utilities for the items. Woeginger \cite{Woe2} also gave an FPTAS for the case when the number of agents, $m$, is a constant. The first non-trivial approximation algorithm for the general \MMA problem is due to Bezakova and Dani \cite{BD} who gave a $(n-m+1)$-approximation algorithm. They also showed the problem is NP-hard to approximate up to any factor smaller than $2$.

Bansal and Sviridenko \cite{BS} introduced a {\em restricted version} of the \MMA problem, called the
{\em Santa Claus} problem, where each item has the property that it has the same utility for a subset of agents and $0$ for the rest. In other words, for each agent $A$ and item $i$, either $u_{A,i} = u_i$, depending only on the item $i$, or $u_{A,i}=0$.
 They proposed an LP relaxation for the problem, referred to as the {\em configuration LP}, and used it to give an $O(\log\log m/\log\log\log m)$-approximation for the Santa Claus problem. Subsequently,
Feige~\cite{Fei08} showed a constant upper bound on the integrality gap of configuration LP for the Santa Claus problem. However his proof is non-constructive and does not translate into an approximation algorithm. Subsequently, Asadpour, Feige and Saberi~\cite{AFS} provided an alternative non-constructive proof of a factor-$5$ upper bound on the integrality gap of the configuration LP.

As for the general \MMA problem, Bansal and Sviridenko \cite{BS} showed that the configuration LP has an integrality gap of $\Omega(\sqrt{n})$ in this setting.
Recently, Asadpour and Saberi \cite{AS} gave an $O(\sqrt{m}\log^3 m)$ approximation for the problem using the same LP relaxation. This is the best approximation algorithm known for the problem prior to our work, while the best current hardness of approximation factor is $2$~\cite{BD}. The main result of our paper is  an $\tilde{O}(n^\eps)$-approximation algorithm for any $\eps=\Omega(\log\log n/\log n)$ for the general \MMA problem, whose running time is $n^{O(1/\eps)}$. This implies a quasi-polynomial time poly-logarithmic approximation to the general \MMA problem. \\
\\
\noindent
Additionally, we investigate a special case of \MMA when each item has positive utility for at most {\em two} agents. We call this special case the {\em $2$-restricted} \MMA problem. When the two positive utilities are identical for both agents, we call the instance a {\em uniform} $2$-restricted instance.
The (uniform) $2$-restricted \MMA reduces to an {\em orientation} problem in  (uniformly) non-uniformly weighted graphs where one has to orient the edges so as to maximize the minimum weighted in-degree of a vertex (a non-uniformly weighted graph has two weights per edge - one for each end point). This orientation problem is called the {\em graph balancing} problem
and is motivated by the min-max analogue studied recently by Ebenlendr et.al. \cite{EKS}.
To the best of our knowledge, prior to our work, the approximability status of the $2$-restricted \MMA problem has been the same as that of the the general \MMA; for the uniform $2$-restricted \MMA the algorithm and analysis of Bansal and Sviridenko for the Santa Claus problem implies a factor-$4$ approximation.
We show that even the uniform $2$-restricted \MMA is hard to approximate up to any factor better than $2$. Moreover, we give a polynomial time $(2+\eps)$-approximation algorithm for the non-uniform $2$-restriced \MMA, for any $\eps > 0$. In fact, we show that the integrality gap of the configuration LP is exactly $2$ in this special case -- the extra $\eps$ in the approximation factor comes from the fact that the configuration LP can only be solved approximately.

{\bf Remark:}
 We have recently learned that independently of our work, Bateni, Charikar, and Guruswami \cite{BCG} showed an approximation algorithm for a special case of the \MMA problem, where the configuration LP is known to have a large integrality gap. Their algorithm achieves an $O(m^\eps)$-approximation in $m^{O(1/\eps)}$ time and a polylogarithmic approximation in quasi-polynomial time for the special case. They also obtain a factor $4$ approximation for the $2$-restricted \MMA.

\subsection{Overview of Results and Techniques}
Our main result is as follows:

\begin{theorem}\label{thm:main}
For any $\eps \geq 8\log\log n/\log n$, there is an $n^{O(1/\eps)}$-time algorithm to compute an
$\tilde{O}(n^{\eps})$-approximate solution for the \MMA problem.
\end{theorem}

We now give an overview of the proof of Theorem \ref{thm:main} and highlight the main ideas involved. We begin by guessing  the value $M$ of the optimal solution. This can be done since, as we show later, we can assume that all utilities are polynomially bounded, losing a constant factor in the approximation ratio.
We then show that we can convert an instance of the general \MMA to what we call
{\em canonical instances}, while losing an $O(\log n)$ factor in the approximation ratio. In a canonical instance, there are only two types of agents -- heavy agents whose utilities for items are either $0$ or $M$, and light agents. Each light agent has a distinct heavy item for which it has utility $M$, and for every other item, the utility is either $M/t$ or $0$, where $t$ is a large integer (larger than the desired approximation factor).
The items with a utility of $M/t$ for a light agent are referred to as the {\em light items}.

Our second idea is that of transforming the problem of assigning items to agents into a network flow problem by the means of {\em private items}. Each agent is assigned at most one distinct private item, for which it has utility $M$.  Note that necessarily the private item of a light agent will be its heavy item. Nevertheless there could be many ways of assigning private items to heavy agents. Of course if we find an assignment of private items for every agent we are trivially done, since this assignment induces a near-optimal solution. Thus the interesting case is when for a {\em maximal} assignment of private items we have some agents who are not assigned any private items. Such agents will be called \emph{terminals}. Suppose for the time being that all light agents are assigned private items.
The key observation which we use is that if the optimum value is $M$, then, given any assignment of private items, there must exist a way of {\em re-assigning} private items such that every terminal is assigned a heavy item. Re-assignment means a heavy agent ``frees'' its private item if it gets another heavy item while a light agent frees its private item if it gets $t$ light items, and then these private items can be re-assigned.

Thus, given any allocation of private items, we can construct a flow network with the property that there exists an {\em integral} flow satisfying certain constraints (for instance, out-flow of $1$ for light agents implies in-flow of $t$). We then design a linear programming relaxation to obtain a fractional flow
solution in the above network. Our LP relaxation has size $n^{O(1/\eps)}$ when the desired approximation ratio is $\tilde{O}(n^\eps)$. However, the integrality gap of the LP relaxation is $\Omega(\sqrt{m})$ -- that is, there is an instance and allocation of private items such that in the resulting network a flow of $M$ can be sent to the terminals fractionally, while an integral flow is not possible even if the ``constraints'' on the flow are ``relaxed'' by a factor of $O(\sqrt{m})$. Thus, rounding the LP-solution directly cannot give a better than $O(\sqrt{m})$ approximation factor.

This brings us to another key technical contribution. Although the LP has a large gap, we show that we can obtain a better approximation algorithm by performing LP-rounding in phases. In each phase we solve the LP and run a rounding algorithm to obtain a solution which is {\em almost feasible}, that is, all terminals get heavy items but some items might be allocated twice.
From this almost-feasible solution, we recover a new assignment of private items and hence a new instance of the LP, one
that has much smaller number of terminals than the stating instance.
 We thus show that in $\mbox{poly}(1/\eps)$ phases, either one of these instances will certify infeasibility of the integral optimum, or we will get an allocation of items which is an $\tilde{O}(n^\eps)$-approximation. This ends the high-level idea of Theorem \ref{thm:main}.\\

 We note that Theorem~\ref{thm:main} can also be used to obtain an approximation in terms of number of agents. In particular, we
 show that we can obtain for any fixed $\eps > 0$, a quasi-polynomial time $m^{\eps}$-approximation algorithm.\\
\\
\noindent
Our second result in the paper is about $2$-restricted \MMA instances. Recall that a $2$-restricted instance is one in which every item $i\in I$ has positive utility for at most two agents $A_i$ and $B_i$. A $2$-restricted \MMA instance is called {\em uniform} if for every item $i$, $u_{A_i,i} = u_{B_i,i}$.
We prove the following theorem which pins down the approximability of $2$-restricted instances.

\begin{theorem}\label{thm:thm2}
For any $\eps > 0$, there exists a $(2+\eps)$-approximation algorithm to the non-uniform $2$-restricted \MMA problem which runs in time $\poly(n,1/\eps)$.
Furthermore, for any $\delta>0$, it is NP-hard to obtain a $(2-\delta)$-approximation algorithm even for the uniform $2$-restricted \MMA problem.
\end{theorem}

In fact, we show that the integrality gap of the configuration LP of~\cite{BS} for $2$-restricted \MMA is bounded from above by $2$ (recall that for general \MMA,
the integrality gap is at least $\Omega{\sqrt{n}}$). As mentioned above, the $2$-restricted \MMA can be cast as an orientation problem
on (non-uniformly) weighted graphs. The main technical lemma which we use for proving the our result above is a generalization of Eulerian orientations to weighted graphs. At a high level, we show that the edges of any (non-uniformly) weighted graph can be oriented such that the total weight coming into any vertex w.r.t the orientation is greater than half of the total weight incident on the vertex in the undirected graph minus the maximum weight edge incident on the vertex. Note that in the case of unweighted graphs, these orientations correspond to Eulerian orientations.

We solve the configuration LP and this gives for each item $i$ a ratio $x_{A_i}$ and $x_{B_i} = 1 - x_{A_i}$ in which it is allocated to the two agents
wanting it. Call an item fractionally divided if both $x_{A_i}$ and $x_{B_i}$ are strictly positive. We show that for any agent the total utility of the set of
items it is allocated fractionally equals $OPT$ plus the maximum utility element in that set. Using the lemma described in the above paragraph, we can get an allocation of value at least $OPT/2$ for every agent.

\subsection{Related Work}
The \MMA problem falls in the broad class of resource allocation problems -- allocating limited resources subject to constraints -- which are ubiquitous in computer science, economics and operations research. When the resources are divisible, the fair allocation problem, also dubbed as {\em cake-cutting problems} have been extensively studied by economists and political scientists with entire books (for example, \cite{BT}) written on the subject. However, the {\em indivisible} case has not received as much attention. There can be many notions of fairness and apart from the notion we study, another measure of fairness studied from an algorithmic point of view has been that of {\em envy-freeness}. An allocation is {\em envy-free} if every agent prefers the set of items allocated to it to a set allocated to any other agent. Lipton et.al. \cite{LMMS} studied the notion of {\em envy} of an allocation and gave polynomial time algorithms to find an allocation with an absolute bound on the envy. We should remark here that the two notions of fairness are not related -- an envy free allocation could be far from being a max-min allocation, and vice-versa.

The complexity of resource allocation problems also depends on the complexity of the utility functions of agents. As we mention above, the utility functions we deal with in this paper are additive -- for every agent $A$, the total utility of a set $S$ of items is simply $u_A(S) := \sum_{i\in S} u_{A,i}$. However, such an assumption on utilities is too restrictive and more general utility functions have been studied in the literature. Two such utilities are {\em submodular utilities} -- for every agent $A$ and any two subsets $S,T$ of items, $u_A(S) + u_A(T) \ge u_A(S\cup T) + u_A(S\cap T)$, and {\em sub-additive utilities} -- $u_A(S) + u_A(T) \ge u_A(S\cup T)$. Note that submodular utilities are a special case of the sub-additive utilities. Khot and Ponnuswami \cite{KP07} gave a $(2m - 1)$-approximate algorithm for \MMA with sub-additive utilities. Recently, Goemans et.al. \cite{GHIM} using the Asadpour-Saberi \cite{AS}
$\tilde{O}(\sqrt{m})$-algorithm as a black box gave a $\tilde{O}(\sqrt{n} m^{1/4})$-approximation for \MMA with submodular utilities. We note that using our main theorem above, the algorithm of \cite{GHIM} gives a $\tilde{O}(n^{1/2+\epsilon})$-approximation for submodular \MMA in time $n^{O(1/\epsilon)}$. We remark here that nothing better than the factor $2$ hardness is known for \MMA even with the general sub-additive utilities.

As we have already mentioned, \MMA may be viewed as a dual problem to the minimum makespan machine scheduling. Lenstra, Shmoys and Tardos \cite{LST} gave a factor $2$-approximation algorithm for the problem and also showed the problem is NP-hard to approximate to a factor better than $3/2$. Closing this gap has been one of the challenging problems in the field of approximation algorithms. Recently Ebenlendr et.al. \cite{EKS} studied a very restricted setting where each job can only go to two machines and moreover takes the same time on both. For this case they gave a $7/4$ approximation algorithm and also showed even this special case is NP-hard to approximate better than a factor $3/2$. Our investigation of the $2$-restricted \MMA is inspired by \cite{EKS} and our hardness result is similar. However, the ideas behind our approximation algorithm are quite different.

\section{Preliminaries}

Our goal is to produce an $\tilde{O}\left (n^{\epsilon}\right )$-approximate solution in time $n^{O(1/\epsilon)}$.
We assume that $n^{\epsilon}\geq \Omega(\log^8 n)$, and thus $\epsilon\geq \Omega(\log\log n/\log n)$.
We also assume that $n$ is larger than any constant and throughout when we say ``$n$ large enough'' we imply larger than a suitable
constant.
We use $M$ to denote the (guessed) value of the optimal solution. If $M\leq \opt$ then our algorithm produces an $\tilde{O}(n^{\eps})$-approximation, otherwise it returns a certificate that $M>\opt$.
For an agent $A$ and item $i$, we use interchangeably the phrases
``$A$ has utility $\gamma$ for item $i$'' and ``item $i$ has utility $\gamma$ for $A$'' to indicate that
$u_{A,i} = \gamma$. We say that an item $i$ is \emph{wanted} by agent $A$ iff $u_{A,i}>0$.

\subsection{Polynomially Bounded Utilities}\label{sec:polybound}

We give a simple transformation that ensures that each non-zero utility value is between $1$ and $2n$, with at most a
factor $2$ loss in the optimal value.
We can assume w.l.o.g. that any non-zero utility value is at least $1$ (otherwise, we can scale up all utilities appropriately),
and that the maximum utility is at most $M$ (the optimal solution value).
For each agent $A$ and item $i$, we define its new utility as follows. If $u_{A,i}<M/2n$ then $u'_{A,i}=0$; otherwise
$$u'_{A,i} =  u_{A,i} \cdot\frac{2n}{M} .$$
Since the optimal solution value in the original instance is $M$, the optimal solution value in the new instance at most $2n$. Moreover, it is easy to see that this value is at least $n$: consider any agent $A$ and the subset $S$ of items assigned to $A$ by $\opt$. The total utility of $S$ for $A$ is at least $M$, and at least $M/2$ of the utility is received from items $i$ for which $u_{A,i}\geq M/2n$. Therefore, the new utility of set $S$ for $A$ is at least $n$.

It is easy to see that any $\alpha$-approximate solution to the transformed instance implies a $(2\alpha)$-approximate solution to the original instance. Let $M'\leq 2n$ be the maximum utility in the transformed instance.
From here on, we assume that our starting instance is the transformed instance with polynomially bounded utilities,
and will denote $M'$ by $M$ and the new utilities $u'_{A,i}$ by $u_{A,i}$.

\subsection{Canonical Instances}

It will be convenient to work with a structured class of instances that we refer to as \emph{canonical instances}.
The notion of a canonical instance depends on the approximation ratio that we desire.
Given any $\epsilon \geq \Omega(\log\log n/\log n)$, we say an instance $\iset$ of \MMA is \emph{$\epsilon$-canonical}, or simply, {\em canonical}  iff:

\begin{itemize}
\item All agents can be partitioned into two sets, namely, a set $L$ of light agents and a set $H$ of heavy agents.

\item Each heavy agent $A\in H$ is associated with a subset $\Gamma(A)$ of items such that each item
in $\Gamma(A)$ has a utility of $M$ for $A$.

\item Each light agent $A\in L$ is associated with

\begin{itemize}
\item
a {\em distinct} item $h(A)$ that has utility $M$ for $A$ and is referred to as the {\em heavy item} for $A$. Note that if $A\neq A'$ then $h(A)\neq h(A')$,

\item
a parameter $N_A \geq n^{\epsilon}$, and

\item
a set $S(A)$ of items such that each item in $S(A)$ has a utility of $M/N_A$ for $A$.
\end{itemize}

\end{itemize}

Given an assignment of items to agents in the canonical solution, we say that a heavy agent $A$ is {\em satisfied} iff it is assigned one of the items in $\Gamma(A)$, and we say that a light agent $A$ is {\em $\alpha$-satisfied} (for some $\alpha\geq 1$) iff it is either assigned item $h(A)$, or it is assigned at least $N_A/\alpha$ items from the set $S(A)$. In the latter case we say that $A$ is  \emph{satisfied by light items}.  A solution is called {\em $\alpha$-approximate} iff all heavy agents are satisfied and all light agents are $\alpha$-satisfied. Given a canonical instance, our goal is to find an assignment of items to agents that $1$-satisfies all agents.

\begin{lemma}
\label{lem:canonical}
Given an instance $\iset$ of the \MMA problem with optimal solution value $M$, we can produce in polynomial time a canonical instance $\iset'$ such that
$\iset'$ has a solution that $1$-satisfies all agents, and any $\alpha$-approximate solution to $\iset'$ can be converted into a $\max\set{O(\alpha\log n),O(n^{\eps}\log n)}$-approximate solution to $\iset$.
\end{lemma}

\begin{proof}
Given an instance $\iset$ of the \MMA problem, we create a canonical
instance $\iset'$ as follows. Define $s=\lfloor \log(M/(n^\eps\log n))\rfloor$. Recall that $M\leq 2n$, so $s \le \log n$ for large enough $n$.
For each agent in $\iset$, the canonical instance $\iset'$ will contain $2s+1$ new agents,
for a total of $m(2s+1)$ agents.
Let $X$ be the set of items in $\iset$. The set $X'$ of items for the instance
$\iset'$ will contain items of $X$ as well as $m(2s)$ additional items that we define later.

Specifically, for each agent $B$ in $\iset$, we create the following collection of new agents and items:
\begin{itemize}
\item
A heavy agent $\chi_0(B)$ and $s$ light agents $\lambda_1(B),\ldots,\lambda_s(B)$ where
each light agent $\lambda_j(B)$ is associated with value $N_{\lambda_j(B)}=M/(s\cdot 2^j) \geq M/(s\cdot 2^s) \ge n^\eps$.

\item
For each $j \in \set{1,\ldots,s}$, if the utility of item $i\in X$ for $B$ is $2^{j-1}<u_{B,i}\leq 2^j$, then agent $\lambda_j(B)$ has utility $s\cdot 2^j=M/N_{\lambda_j(B)}$ for $i$. If $i\in X$ is an item for which $u_{B,i}>2^s$, then $\chi_0(B)$ has utility $M$ for item $i$.

\item
Additionally, for each light agent $\lambda_j(B)$ there is a heavy item $h(\lambda_j(B))$ and a heavy agent $\chi_j(B)$. Item $h(\lambda_j(B))$
has utility $M$ for both $\lambda_j(B)$ and $\chi_j(B)$.

\item
Finally, we have a set of $s$ items $Y_B=\set{i_1(B),\ldots,i_s(B)}$ such that each item in $Y_B$ has a utility of $M$ for
each of the agents in $\{ \chi_0(B), \chi_1(B), \ldots, \chi_s(B) \}$, the set of heavy agents for $B$.
\end{itemize}

This completes the definition of the canonical instance $\iset'$. Consider an optimal solution to $\iset$. We show an assignment that $1$-satisfies all agents. Consider some agent $B$ in the transformed instance, and let $T(B)$ be the set of items assigned to $B$ in the solution to $\iset$. We can partition $T(B)$ into $(s+1)$ sets as follows. The set $T_0\sse T(B)$ contains all items $i$ with $u_{B,i}> 2^s$. For each $j \in \set{1,\ldots,s}$, the set $T_j$ contains all items $i$ with $2^{j-1}<u_{B,i}\leq 2^j$. Assume first that $T_0\neq \emptyset$, and let $i$ be any item in $T_0$. Then we assign item $i$ to heavy agent $\chi_0(B)$. The remaining $s$ heavy agents corresponding to $B$ now get assigned one item from $Y_B$ each. The light agents $\lambda_j(B)$ are assigned their heavy items $h(\lambda_j(B))$. All agents corresponding to $B$ are now $1$-satisfied.
Assume now that $T_0=\emptyset$. Then there is $j \in \set{1,\ldots,s}$, such that the utility of items in $T_j$ is at least $M/s$ for $B$, so $|T_j|\geq  M/(s\cdot 2^{j})=N_{\lambda_j(B)}$. We assign all items in $T_j$ to $\lambda_j(B)$, and $h(\lambda_j(B))$ is assigned to the heavy agent $\chi_j(B)$. Now the remaining $s$ heavy agents are assigned one item of $Y_B$ each. For each one of the remaining light agents, we assign $h(\lambda_{j'}(B))$ to $\lambda_{j'}(B)$.
Therefore, the canonical instance has a solution that $1$-satisfies all agents.

Conversely, consider now any $\alpha$-approximate solution for the canonical instance $\iset'$. Let $B$ be some agent in the original instance. Consider the corresponding set of heavy agents $\chi_0(B),\ldots,\chi_s(B)$. Since there are only $s$ items in the set $Y_B$, at least one of the heavy agents is not assigned an item from this set. Assume first that it is the heavy agent
$\chi_0(B)$. Then it must be assigned some item $i\in X$ for which agent $B$ has utility at least $2^s\geq M/(2n^\eps\log n)$. We then assign item $i$ to agent $B$. Otherwise, assume it is $\chi_j(B)$, $j\neq 0$ that is not assigned an item from $Y_B$. Then $\chi_j(B)$ is assigned item $h(\lambda_j(B))$, and so $\lambda_j(B)$ must be assigned a set $S'$ of  at least $N_{\lambda_j(B)}/\alpha=M/(s\cdot 2^j\cdot \alpha)$ items, each of which has a utility of at least $2^{j-1}$ for $B$. We then assign the items in $S'$ to $B$. Since $s \le \log n)$, we obtain a $\max\set{O(n^{\epsilon}\log n),O(\alpha\log n)}$-approximate solution.
\end{proof}

From now on we focus on finding an approximate solution to the canonical instance. We assume that the optimal solution can $1$-satisfy all agents. From the above claim, a solution that $\alpha$-approximates all agents in the canonical instance will imply a $\max\set{O(n^{\epsilon}\log n),O(\alpha\log n)}$-approximate solution to the original instance.

\subsection{Private Items and Flow Solutions}\label{sec:private}
One of the basic concepts of our algorithm is that of private items and flow solutions defined by them. Throughout the algorithm we maintain an assignment $P$ of private items to agents. Such an assignment is called \emph{good} iff it satisfies the following properties:

\begin{itemize}

\item For every light agent $A\in L$, its private item is $P(A)=h(A)$ (that is, the distinct item of utility $M$ associated with the light agent $A$).

\item An item can be a private item for at most one agent. The set of items that do not serve as private items is denoted by $S$.

\item The set of heavy agents that have private items are denoted by $H'$. The remaining heavy agents are called \emph{terminals} and are denoted by $T$. Item $i$ can be a private item of heavy agent $A\in H$ only if $i\in \Gamma(A)$
    (recall that $\Gamma(A)$ is the set of items for which $A$ has utility $M$).
\end{itemize}

The initial assignment $P$ of private items to heavy agents is obtained as follows. We create a bipartite graph $G=(U,V,E)$, where $U=H$, $V$ is the set of items that do not serve as private items for light agents, and $E$ contains an edge between $A\in U$ and $i\in V$ iff $i\in \Gamma(A)$. We compute a maximum matching in $G$ that determines the assignment of private items to heavy agents. To simplify notation, we say that for a terminal $A\in T$, $P(A)$ is undefined and $\set{P(A)}\triangleq \emptyset$.

\paragraph{The Flow Network}
Given a canonical instance $\iset$ and an assignment $P: L\cup H'\rightarrow I$ of private items, we define the corresponding {\bf directed} flow network $N(\iset,P)$ as follows.
The set of vertices is $\aset\cup I\cup\set{s}$. Source $s$ connects to every vertex $i$ where $i\in S$. If agent $A\in \aset$ has a private item and $i=P(A)$, then vertex $A$ connects to vertex $i$.  If $A$ is a heavy agent and $i\in \Gamma(A)\setminus\set{P(A)}$, then vertex $i$ connects to vertex $A$. If $A$ is a light agent and $i\in S(A)$ then vertex $i$ connects to vertex $A$.
Let $N(\iset,P)$ denote the resulting network.
A feasible integral flow in this is a flow obeying the following constraints.

\renewcommand{\theenumi}{C\arabic{enumi}}

\begin{enumerate}

\item All flow originates at the source $s$. \label{C2}

\item Each terminal agent $A\in T$ receives one flow unit. \label{C3}

\item For each heavy agent $A\in H$, if the incoming flow is $1$ then the outgoing flow is $1$; otherwise both are $0$. \label{C4}

\item For each item $i\in I$, if the incoming flow is $1$ then the outgoing flow is $1$; otherwise both are $0$.\label{C5}

\item For each light agent $A\in L$, if the incoming flow is $N_A$ then the outgoing flow is $1$; otherwise both are $0$. \label{C6-Light-constraint}

\end{enumerate}

An {\em integral flow} obeying the above conditions is called a \emph{feasible flow}.

\begin{lemma}
An optimal integral solution to the canonical instance $\iset$ gives a feasible flow in $N(\iset,P)$.
\end{lemma}
\begin{proof}
Consider the optimal solution. We assume w.l.o.g. that all items in $I\setminus S$ are assigned: otherwise if $i\in I\setminus S$ is not assigned by the solution, we can assign it to the unique agent $A\in\aset$ such that $P(A)=i$. Let $A\in H$ be a heavy agent. If it is assigned an item $i\neq P(A)$, then we send one flow unit from vertex $i$ to vertex $A$. If $A\not\in T$, then it also sends one flow unit to $P(A)$. Consider now a light agent $A\in L$. If it is not assigned $P(A)$, then there is a collection $S'\sse S(A)$ of $N_A$ items assigned to $A$. Each of these items sends one flow unit to $A$, and $A$ sends one flow unit to $P(A)$. Finally, each item $i\in S$ that participates in the assignment receives one flow unit from $s$. It is easy to see that this is a feasible flow.
\end{proof}

We say that a flow is {\em $\alpha$-feasible} iff constraints (\ref{C2})--(\ref{C5}) hold for it, and Constraint (\ref{C6-Light-constraint}) is replaced by the following relaxed constraint:

\begin{enumerate}
\setcounter{enumi}{5}
\item  For each light agent $A\in L$, if the incoming flow is at least $N_A/\alpha$ then the outgoing flow is $1$; otherwise both are $0$.\label{C6-Light-constraint-relaxed}
\end{enumerate}

\begin{lemma} An integral $\alpha$-feasible flow in $N(\iset,P)$ gives an $\alpha$-approximate solution for the canonical instance $\iset$.\end{lemma}
\begin{proof}
Consider an integral $\alpha$-feasible flow. Consider any agent $A$ that may be heavy or light.
We simply assign it every item that sends flow to it if there is such an item. If there is no such item, we assign $P(A)$ to $A$. It is easy to verify that this is an $\alpha$-approximate solution, since every terminal receives one flow unit and all other agents that do not have any flow going through them can be assigned their private items.
\end{proof}

Let $I^*$ be the set of items and $H^*$ the set of heavy agents reachable from $s$ by paths that do not contain light agents. A useful property of our initial assignment of private items is that $H^*$ does not contain any terminals (otherwise we could have increased the matching). Throughout the algorithm, the assignment of private items to $H^*$ does not change, and the above property is preserved.
Given any pair $v,v'$ of vertices and any path $p$ that starts at $v$ and ends at $v'$, we say that $p$ connects $v$ to $v'$  \emph{directly} it does not contain any intermediate vertices representing light agents (we allow $v$ and $v'$ to be light agents under this definition).  We say that $v$ is {\em directly connected} to $v'$ if such a path exists. Similarly, given an integral flow solution, we say that $v$ sends flow directly to $v'$ iff flow is being sent via path $p$ that does not contain light agents as intermediate vertices.

\paragraph{An Equivalent Formulation:}
A flow-path $p$ is called a \emph{simple path} iff it does not contain any light agents as intermediate vertices, and either a) starts and ends at light agents, or b) starts at a light agent and ends at a terminal, or c) starts at the source $s$ and ends at a light agent.
The problem of finding an integral flow satisfying the properties \ref{C2}--\ref{C5} and \ref{C6-Light-constraint-relaxed},
is {\em equivalent} to finding a collection of simple paths $\pset$ such that:

\renewcommand{\theenumi}{\^{C}\arabic{enumi}}

\begin{enumerate}
\item All paths in $\pset$ are internally-disjoint (i.e. they do not share intermediate vertices). \label{property 1 simple paths}

\item Each terminal has exactly one path $p\in \pset$ terminating at it.

\item For each light agent $A\in L$, there is at most one path $p\in \pset$ starting at $A$, and if such a path exists,
then there are at least $N_A/\alpha$ paths in $\pset$ terminating at $A$. \label{property last simple paths}
\end{enumerate}

\subsection{Structured Near-Optimal Solutions for Canonical Instances}

Given a canonical instance $\iset$ and an assignment $P$ of private items, an optimal solution $\opt$ to $\iset$ defines a feasible integral flow in $N(\iset,P)$. Consider graph $G$ obtained from $N(\iset,P)$ as follows: we remove the source $s$ and all its adjacent edges. We also remove all edges that do not carry flow in $\opt$.
Graph $G$ is then a collection of disjoint trees. Some of the trees in this collection are simply isolated vertices that correspond to agents  and items with no flow passing through them. Such agents are assigned their private items. We ignore these trivial trees in this section and focus only on trees with $2$ or more vertices.

Each (non-trivial) tree $\tau$ in the collection has a terminal $t\in T$ at its root. Each heavy agent in the tree has one incoming and one outgoing edge, and each light agent $A$ has $N_A$ incoming edges and one outgoing edge.  The leaves are items in $S$. Such a solution is called an {\em $h$-layered forest} iff for every tree $\tau$, the number of light agents on any leaf-to-root path is the same and is bounded by $h$; we denote this number by $h(\tau)$. Note that $h(\tau)\geq 1$ must hold since there are no direct paths between $s$ and the terminals, and thus $1\leq h(\tau)\leq h$ for all $\tau$. It is convenient to work with layered solutions, and we now show that for any canonical instance,
there exists an $(h+1)$-approximate $h$-layered solution for $h=8/\epsilon$.

\begin{lemma} \label{lemma: h-level solution} There is an $(h+1)$-approximate solution $\sset$ to a canonical instance instance $\iset$, that defines an $h$-layered forest, for $h=8/\epsilon$.
\end{lemma}
\begin{proof}
We will start with an optimal solution $\opt$ for $\iset$, and convert it into an $h$-layered solution in which every light agent will be $(h+1)$-satisfied.
Consider any tree $\tau$ in the collection of disjoint trees corresponding to $\opt$. We convert it into an $h$-layered tree in $h$ iterations.

A light agent $A\in L$ that belongs to $\tau$ is called a \emph{level-$1$ agent} iff it receives at least $N_A/(h+1)$ flow units directly from items in $S$. Let $L_1(\tau)$ be the set of all level-1 light agents of $\tau$. Consider some agent $A\in L_1(\tau)$. For each child $v$ of $A$ in $\tau$, if $v$ is not on a simple path $p$ connecting an item of $S$ to $A$, then we remove $v$ together with its sub-tree from $\tau$.

In general, for $j>1$, a light agent $A$ is a {\em level-$j$ agent} iff it receives at least $N_A/(h+1)$ flow units directly from level-$(j-1)$ agents.
In iteration $j$, we consider the set $L_j(\tau)$ of level-$j$ agents. Let $A\in L_j(\tau)$ be any such agent, and
let $v$ be any child of $A$ in $\tau$. If $v$ lies on a simple path $p$ connecting some level-$(j-1)$ agent to $A$ in $\tau$, then we do nothing. Otherwise, we remove $v$ and its subtree from $\tau$. We claim that after iteration $h$ is completed, all remaining light agents in $\tau$ belong to $\lset=\cup_{j=1}^hL_j(\tau)$. Thus we can make every $\tau$ $h$-layered implying a $(h+1)$-approximate $h$-layered solution.

Assume otherwise. Then we can find a light agent $A\notin \lset$ in $\tau$, such that all light agents in the sub-tree of $A$ belong to $\lset$ (Start with arbitrary $A\not\in\lset$. If there is a light agent $A'$ in the subtree of $A$ that does not belong to $\lset$, then continue with $A'$. When this process stops we have agent $A$ as required). Agent $A$ receives at least $N_A$ flow units in $\tau$, but it does not belong to $L_j(\tau)$ for $1\leq j\leq h$. That is, it receives less than $N_A/(h+1)$ flow units directly from light agents in $L_j(\tau)$ for $1\le j\le h-1$ and less than $N_A/(h+1)$ flow units directly from $S$.
It follows that it must receive at least $N_A/(h+1)\geq n^{\eps}/(h+1)$ units of flow directly from agents in $L_h$. Each one of these agents receives at least $n^{\eps}/(h+1)$ flow from agents in $L_{h-1}(\tau)$ and so on. So for each $j: 1\leq j\leq h$, the sub-tree of $A$ contains at least $(n^{\eps}/(h+1))^{h-j+1}$ agents from $L_j(\tau)$, and in particular it contains $(n^{\eps}/(h+1))^h$ agents from $L_1(\tau)$. We now show that by our choice of $\eps,h$, $(n^{\eps}/(h+1))^h > m$ which would be a contradiction.

Recall that $n^{\eps}\geq \log^8n$ and $1>\eps\geq 8\log\log n/\log n$. So $8 \le (h=8/\eps) \le \log n/\log\log n$.
Thus, for $n$ large enough $(h+1) \le 2h \le 2\log n/(\log\log n) \le \log n \le n^{\eps/8}$ giving us $(n^{\eps}/(h+1))^h \ge (n^{7\eps/8})^h \ge n^7 \ge n > m$.
\end{proof}

From now on we focus on $h$-layered instances. For simplicity, we scale down all values $N_A$ for $A\in\aset$ by the factor of $(h+1)$, so that an optimal $h$-layered solution can $1$-satisfy all agents. Note that this increases the approximation factor of our algorithm by the factor of $(h+1)$.

\subsection{The [BS] Tree Decomposition}
One of the tools we use in our construction is a tree-decomposition theorem of Bansal and Sviridenko~\cite{BS}.
We remark that this theorem has no connection to the trees induced in the flow network $N(\iset,P)$ by a feasible solution. The setup for the theorem is the following. We are given an {\bf undirected} bipartite graph $G=(\aset,I,E)$ where $\aset$ is a set of agents, $I$ is a set of items and $E$ contains an edge $(A,i)$ iff $i$ has utility $M$ for $A$. Additionally, every agent $A$ is associated with a value $0\leq x_A\leq 1$. (We can think of $x_A$ as the extent to which $A$ is satisfied by light items in a fractional solution). We are also given a fractional assignment $y_{A,i}$ of items, such that:

\begin{eqnarray}
\forall i\in I&&\sum_{A\in\aset}y_{A,i}\leq 1 \label{constraint 1 for BS}\\
\forall A\in \aset&&\sum_{(A,i)\in E}y_{A,i}=1-x_A \label{constraint 2 for BS}\\
\forall (A,i)\in E&&0\leq y_{A,i}\leq 1 \label{last constraint for BS}
\end{eqnarray}

The theorem of~\cite{BS} shows that such an assignment can be decomposed into  a more convenient structure. This structure is a collection of disjoint trees in graph $G$. For each such tree $\tau$ the summation of values $x_A$ for agents in $\tau$ is at least $\half$, and moreover if $\aset(\tau)$ is the set of agents of $\tau$ and $I(\tau)$ is the set of items of $\tau$, then for each agent $A\in \aset(\tau)$, it is possible to satisfy all agents in $\aset(\tau)\setminus\set{A}$ by items in $I(\tau)$.

\begin{theorem} (\cite{BS})\label{thm: BS}
There exists a decomposition of $G=(\aset,I,E)$ into a collection of disjoint trees $T_1,T_2,\ldots,T_s$, such that, for each tree $T_j$, either (1) $T_j$ contains a single edge between some item $i$ and some agent $A$, or (2) the degree
of each item $i\in I(T_j)$ is $2$ and $\sum_{A\in\aset(T_j)} x_A > 1/2$.
\end{theorem}

\begin{corollary} For each tree $T_j$ in the decomposition, for each agent $A\in \aset(T_j)$, it is possible to satisfy all agents in $\aset(T_j)\setminus\set{A}$ by items in $I(T_j)$.
\end{corollary}

\begin{proof} Root tree $T_j$ at vertex $A$. Now every agent $A'\neq A$ is assigned item $i$, where $i$ is the parent of $A'$ in $T_j$. Since the degree of every item is at most $2$, this is a feasible assignment.
\end{proof}

\begin{proof} (of Theorem~\ref{thm: BS})
We remove from $E$ all edges $(A,i)\in E$ with $y_{A,i}=0$. Let $E^*\sse E$ be the set of edges $(A,i)$ with $y_{A,i}=1$. We remove from $G$ edges in $E^*$ together with their endpoints.

\noindent
{\bf Step 1}: Converting $G$ into a forest. We will transform values $y_{A,j}$ so that the set of edges with non-negative values $y_{A,j}$ forms a forest, while preserving constraints (\ref{constraint 1 for BS})--(\ref{last constraint for BS}), as follows. Suppose there is a cycle $\cset$ induced by edges of $E$. Since the cycle is
even (the graph being bipartite), it can be decomposed into two matchings $M_1$ and $M_2$. Suppose the smallest value $y_{A,i}$ for any edge  $(A,i)\in M_1\cup M_2$ is $\delta$. For each $(A,i)\in M_2$, we increase $y_{A,i}$ by $\delta$ and for each $(A,i)\in M_1$ we decrease $y_{A,i}$ by $\delta$. It is easy to see that constraints~(\ref{constraint 1 for BS})--(\ref{last constraint for BS}) continue to hold, and at least one edge $(A,i)\in M_1\cup M_2$ has value $y_{A,i}=0$. All edges $(A,i)$ with $y_{A,i}=0$ are then removed from $E$, and all edges $(A,i)$ with $y_{A,i}=1$ are added to $E^*$ with $A$ and $i$ removed from $G$. We can continue this process until no cycles remain in $G$.

We now fix some tree $\tau$ in $G$. While there exists an item $i\in I(\tau)$ of degree $1$, we perform Step 2.

\noindent
{\bf Step 2}:  If there is an item $i$ in $\tau$ with degree $1$, then let $A$ be the agent with $(A,i)\in E$. We set $y_{A,i}=1$ and for all $i'\neq i$, $y_{A,i'}=0$. We then add $(A,i)$ to $E^*$, removing the edge and both its endpoints from $G$. Notice that constraints (\ref{constraint 1 for BS})--(\ref{last constraint for BS}) continue to hold.

Assume now that the degree of every item $i\in I(\tau)$ is $2$. Clearly then $|\aset(\tau)|=|I(\tau)|-1$. Then $\sum_{A\in \aset(\tau)}\sum_{i\in I(\tau)}y_{A,i}\leq |I(\tau)|= |\aset(\tau)|-1$. On the other hand,
$\sum_{A\in \aset(\tau)}\sum_{i\in I(\tau)}y_{A,i}=\sum_{A\in\aset(\tau)}(1-x_A)=|\aset(\tau)|-\sum_{A\in \aset(\tau)}x_A$. Therefore, $\sum_{A\in \aset(\tau)}x_A\geq 1$.

Otherwise, while there is an item in $I(\tau)$ of degree greater than $2$, we perform Step 3:

{\bf Step 3}: Root tree $\tau$ at an arbitrary vertex in $I(\tau)$. Let $i\in I(\tau)$ be a vertex of degree at least $3$, such that in the sub-tree of $i$ all items have degree $2$. Now consider the
children  of $i$ - let them be $A_1,A_2,\ldots, A_r$ with $r \ge 2$. Note that there is $j: 1\leq j\leq r$ with  $y_{A_j,i} < 1/2$. Remove the edge $(A,i)$ from $G$, and add the sub-tree $\tau'$ rooted at $A_j$ to the decomposition. Note that all item vertices in this sub-tree has degree exactly $2$.
Also note that:

$\sum_{A'\in \aset(\tau')}\sum_{i'\in I(\tau')}y_{A',i'}\leq |I(\tau')|= |\aset(\tau')|-1$, while on the other hand
$\sum_{A'\in \aset(\tau')}\sum_{i'\in I(\tau')}y_{A',i'}=\sum_{A'\in\aset(\tau')}(1-x_{A'})-y_{A,i}=|\aset(\tau')|-\sum_{A'\in \aset(\tau')}x_{A'}-y_{A,i}$.

Therefore, $\sum_{A'\in \aset(\tau')}x_{A'}\geq 1-y_{A,i}\geq \half$.
\end{proof}

\section{An $\tilde{O}(\sqrt{n})$-Approximation via $1$-Layered Instances}
\label{sec: 1-layered instances}

In this section, as a warm-up towards the eventual goal of computing optimal multi-layered solutions,
we consider the special case of finding an optimal $1$-layered solution, that is,
we restrict solutions to trees $\tau$ with $h(\tau)$=1.
We design an LP relaxation and a rounding scheme for the relaxation,
to obtain an $O(\log n)$-approximation to the problem of finding the optimal
$1$-layered solution. We then describe how this suffices to obtain an $\tilde{O}(\sqrt{n})$-approximation overall.

As we restrict our solution to trees $\tau$ with $h(\tau)=1$, we know that the items in $I^*$ (items reachable directly from $s$) will only be used to send flow to the light agents, and items in $I'=I\setminus I^*$ will only be used to send flow directly to terminals. Similarly, heavy agents in $H^*$ will send flow to light agents, while heavy agents in $H\setminus H^*$ will send flow directly to terminals. Therefore, if there are any edges from items in $I'$ to agents in $H^*$ we can remove them (no edges are possible between agents in $H\setminus H^*$ and items in $I^*$).

We now proceed with the design of an LP relaxation.
For each light agent $A\in L$, we have a variable $x_A$ showing whether or not $A$ is satisfied as a light agent (by light items). We have a flow of type $A$, we call it $f_A$, and require that $N_A\cdot x_A$ units of this flow are sent to agent $A$, while at most $x_A$ flow units of $f_A$ go through any item in $I^*$. Finally, the total flow through any item of any type is bounded by $1$. Let $\pset(A)$ be the set of paths originating in $s$ and ending at $A$, that only use items in $I^*$. We then have the following constraints:

\begin{eqnarray}
\forall A\in L& \sum_{p\in \pset(A)}f_A(p)=N_A\cdot x_A&\mbox{\quad($N_A\cdot x_A$ flow units sent to $A$)}\\
\forall A\in L,\forall i\in I^*& \sum_{\stackrel{p\in P(A):}{i\in p}}f_A(p)\leq x_A&\mbox{\quad(capacity constraint w.r.t. flow of type $A$)} \label{capacity-constraint}\\
\forall i\in I^*& \sum_{A\in L} \sum_{p: i\in p}f_A(p)\leq 1&\mbox{\quad(general capacity constraint)}
\end{eqnarray}

Additionally, we need to send one flow unit to each terminal. For $A\in L$, $t\in T$, let $\pset(A,t)$ be all paths directly connecting $A$ to $t$ that only use items in $I'$.
We now have the standard flow constraints:

\begin{eqnarray}
\forall t\in T& \sum_{A\in L}\sum_{p\in \pset(A,t)}f(p)=1&\mbox{\quad(Each terminal receives $1$ flow unit)}\\
\forall A\in L& \sum_{t\in T}\sum_{p\in \pset(A,t)}f(p)=x_A&\mbox{\quad(Light agent $A$ sends $x_A$ flow units)}\\
\forall i\in I'&\sum_{p: i\in p}f(p)\leq 1&\mbox{\quad(Capacity constraints)}
\end{eqnarray}

The rounding algorithm consists of three steps. In the first step we perform the BS tree decomposition on the part of the graph induced by $I'$. The second step is randomized rounding which will create logarithmic congestion. In the last step we take care of the congestion.

\paragraph{Step 1: Tree decomposition} We consider the graph induced by vertices corresponding to all agents in set $Z=L\cup (H\setminus H^*)$ and the set $I'$ of items. Notice that agents in $Z$ only have utilities $M$ for items in $I'$.
For each light agent $A\in L$ we have a value $x_A$ (extent to which $A$ is satisfied as light item). For all other agents $A$, let $x_A=0$. Our fractional flow solution can be interpreted as a fractional assignment of items in $I'$ to agents in $Z$, such that  each agent $A\in Z$ is fractionally assigned $(1-x_A)$-fraction of items in $I'$, as follows. Let $A\in H\setminus H^*$ be any heavy agent. For each item $i\in I'$, we set $y_{A,i}$ to be the amount of flow sent on edge $(i\rightarrow A)$ if such an edge exists. If $A$ is a non-terminal agent, let $z$ be the amount of flow sent on edge $(A\rightarrow P(A))$. We set $y_{A,P(A)}=1-z$. For a light agent $A$, we set $y_{A,P(A)}=1-x_A$. It is easy to see that this assignment satisfies Constraints~(\ref{constraint 1 for BS})--(\ref{last constraint for BS}). Therefore, we can apply Theorem~\ref{thm: BS} and obtain a collection $T_1,\ldots,T_s$ of trees together with a matching $\mathcal{M}$ of a subset of items $I''\sse I'$ to a subset $Z'\sse Z$ of heavy agents, that do not participate in trees $T_1,\ldots,T_s$.

\paragraph{Step 2: Randomized rounding}
Consider some tree $T_j$ computed above. We select one of its light agents $A$ with probability $x_A/X$, where $X$ is the summation of values $x_{A'}$ for $A'\in T_j$. Notice that $x_A/X\leq 2x_A$. The selected light agent will eventually be satisfied as a light agent. Once we select one light agent $A_j$ for each tree $T_j$, we can satisfy the remaining agents of the tree with items of the tree. Let $L^*=\set{A_1,\ldots,A_s}$. We have obtained an assignment of an item from $I'$ to every agent in $Z\setminus L^*$. This assignment in turn defines a collection of simple flow-paths $\pset$, where every path $p\in \pset$ connects an agent in $A_1,\ldots,A_s$ to a terminal, with exactly one path leaving each agent $A_j$ and exactly one path entering each terminal, as follows: Consider any agent $A\in Z\setminus\set{A_1,\ldots,A_s}$ and assume that $A$ is assigned some item $i\in I'$. If $i\neq P(A)$, then send one flow unit from $i$ to $A$, and if $A\not \in T$, send one flow unit from $A$ to $P(A)$.

We now turn to finding a collection of simple paths to satisfy the agents in $L^*$. For each $A\in L^*$ we scale its flow $f_A$ by the factor of $1/x_A$, so that $A$ now receives $N_A$ flow units. For agents not in $L^*$, we reduce their flows to $0$. Notice that due to constraint (\ref{capacity-constraint}), at most $1$ unit of flow $f_A$ goes through any item. Since each agent $A$ is selected with probability at most $2x_A$, using the standard Chernoff bound, we get that w.h.p. the congestion (total flow) on any vertex is $O(\log n)$.

\paragraph{Step 3: Final solution}
In the final solution, we require that each agent $A\in L^*$ receives $\lfloor N_A/\log n\rfloor $ flow units {\em integrally} via internally disjoint paths from $s$. Since our original flow has congestion $O(\log n)$, we know that such a {\em fractional flow} exists. By integrality of flow we can get such an integral flow.

We now show that this algorithm is enough to get an $\tilde{O}(\sqrt{n})$ approximation  for max-min allocation.
To obtain such an approximation, it is enough to consider instances where $M\geq 4\sqrt n$. It is easy to see by using a modification of Lemma \ref{lemma: h-level solution} that in this case, by losing a constant factor we can assume that the optimal solution consists of trees $\tau$ with $h(\tau)=1$. Therefore the algorithm presented here will provide an $O(\log n)$-approximation for resulting canonical instances and $\tilde{O}(\sqrt{n})$-approximation overall.

\section{Almost Feasible Solutions for Multi-Layered Instances}\label{sec: almost-feasible-solutions}
\label{sec: almost feasible solutions}

We now generalize the algorithm from the previous section to arbitrary number of layers. From Lemma~\ref{lemma: h-level solution}, it is enough to use $h=8/\eps$ layers  (recall that $1/\eps = O(\log n/\log\log n)$).

There is a natural generalization of the algorithm from the previous section, where we write a multi-layered flow LP, and then perform $h$ rounds of randomized rounding, starting from the last layer. This is what the algorithm presented here does. However we are unable to obtain a feasible solution. Instead we obtain an almost-feasible solution in the following sense: all constraints~(\ref{property 1 simple paths})--(\ref{property last simple paths}) hold, except that for some items and some heavy agents there will be two simple paths containing them: one terminating at some light agent and one at a terminal. Similarly, some light agents $A$ will have two simple paths starting at $A$, one terminating at another light agent and one at a terminal. The fact that we are unable to obtain a feasible solution is not surprising: the LP that we construct has an $\Omega(\sqrt{m})$ integrality gap, as we show in Section~\ref{sec:int-gap} in the Appendix. Surprisingly we manage to bypass this obstacle in our final algorithm presented in Section~\ref{sec: final-algorithm}, and obtain a better approximation guarantee while using the LP-rounding algorithm from this section as a subroutine.

More formally, in this section our goal is to prove the following theorem.

\begin{theorem}~\label{thm: almost feasible solutions}
Let $\iset=(\aset,I)$ be any $1$-satisfiable canonical instance and let $P$ be a good assignment of private items to non-terminal agents. Let $N(\iset,P)$ be the corresponding flow network and let $\alpha=O(h^4\log n)$. Then we can find, in polynomial time, two collections $\pset_1$ and $\pset_2$ of simple paths with the following properties.

\begin{properties}{D}
\item All paths in $\pset_1$ terminate at the terminals and all paths in $\pset_2$ terminate at light agents. Moreover each terminal lies on some path in $\pset_1$. \label{property almost-feasible-first}

\item All paths in $\pset_1$ are completely vertex disjoint, and paths in $\pset_2$ are internally vertex-disjoint but they may share endpoints. A non-terminal agent or an item may appear both in $\pset_1$ and in $\pset_2$.

\item For each light agent $A$, there is at most one path in $\pset_1$ and at most one path in $\pset_2$ that
originates at $A$ (so in total there may be two paths in $\pset_1\cup \pset_2$ originating at $A$).

\item If there is a path $p\in \pset_1\cup \pset_2$ originating at some light agent $A\in L$, then there are at least $N_A/\alpha$ paths in $\pset_2$ that terminate in $A$. \label{property almost-feasible-last}
\end{properties}
\end{theorem}

Our LP itself is defined on a new graph $N_h(\iset,P)$, that can be viewed as a ``structured'' or ``layered'' version of $N(\iset,P)$. We start by describing the graph $N_h(\iset,P)$. After that we define the LP itself, and finally the randomized rounding algorithm.

\subsection{The Graph $N_h(\iset,P)$}

Recall that any integral solution induces a collection of disjoint trees $\tau$ with various heights $h(\tau)\leq h$. To simplify the description of the LP (and at the cost of losing another factor $h$ in the approximation ratio), our graph will consist of $h$ subgraphs, where subgraph $G_{h'}$, $1\leq h'\leq h$, is an $h'$-layered graph that will correspond to trees of height $h(\tau)=h'$.

We start with the description of $G_{h'}$.
We create $h'$ copies of the set of light agents: one copy for each layer. From now on we will treat these copies as distinct agents. Let $L_1,\ldots,L_{h'}$ denote the sets of agents corresponding to the layers.
The graph consists of $h'$ levels, where level $j$ starts with light agents in $L_{j-1}$ and ends with light agents in $L_j$ (the first level starts with the source $s$ and ends with $L_1$). We now proceed to define each level.

Level $1$ contains the source $s$, a copy of each item in $I^*$ and a copy of each heavy agent in $H^*$ (recall that these are the items and the agents that can be reached directly from $s$). The assignment of private items is as before. There is an edge from $s$ to every item in $S$, an edge from every heavy agent $A$ to $P(A)$, and an edge from each item $i\in \Gamma(A)\setminus P(A)$ to $A$. Additionally, if item $i$ is a light item for a light agent $A$, we put an edge between $i$ and a copy of $A$ in $L_1$.

Level $j$ is defined as follows. It contains a copy of each heavy non-terminal agent $A\in H\setminus T$ and a copy of each item $i\in I\setminus S$. Let $H_j$ and $I_j$ denote the set of the heavy agents and items at level $j$. Recall that each item $i\in I\setminus S$ is a private item of some agent. Consider some such item $i$. If it is a private item of a heavy agent $A$, then we add an edge from a copy of $A$ in $H_j$ to a copy of $i$ in $I_j$. If it is a private item of light agent $A$, then we add an edge from a copy of $A$ in $L_{j-1}$ to a copy of $i$ in $I_j$. If $i$ is a non-private item admissible for heavy agent $A$, then we add an edge from a copy of $i$ in $I_j$ to a copy of $A$ in $H_j$. If $i$ is a light item for a light agent $A$ then we add an edge from a copy of $i$ in $I_j$ to a copy of $A$ in $L_j$.
This completes the description of $G_{h'}$.
We will denote by $H_j^{h'},I_j^{h'}$ the heavy agents and items at level $j$ of $G_{h'}$, and we will omit the superscript $h'$ when clear from context. Our final graph consists of the union of graph $G_{h'}$, for $1\leq h'\leq h$. Finally, we add an additional part to our graph that will be used to route flow directly to terminals.

This part consists of a set $\hat{H}$ containing a copy of every heavy agent and a set $\hat{I}$ containing a copy of every item $i\in I\setminus S$.
Note that every item in $\hat{I}$ is a private item of some agent.
If $i\in \hat{I}$ is a private item of a heavy agent $A$ then we add an edge from a copy of $A$ in $\hat{H}$ to a copy of $i$ in $\hat{I}$. If it is a private item of a light agent $A$, then for all $h': 1\leq h'\leq h$, we add an edge from the copy of $A$ in $L_{h'}^{h'}$ to the copy of $i$ in $\hat{I}$ (so there are only edges from the last layer of each $G_{h'}$ to items in $\hat{I}$). If $i$ is an admissible but non-private item for some heavy agent $A$ then we add an edge from the copy of $i$ in $\hat{I}$ to the copy of $A$ in $\hat{H}$.

This completes the description of the graph $N_h(\iset,P)$. Notice that for each item $i\in I$ and each agent $A\in \aset$ of $\iset$, there are at most $h^2$ copies in $N_h(\iset,P)$. We will be looking for an integral flow in this graph satisfying conditions (\ref{property 1 simple paths})--(\ref{property last simple paths}). Notice that given an optimal solution to the original instance, it is easy to convert it into a feasible integral flow in the new instance. Moreover, in any feasible integral flow solution to the new instance, each tree $\tau$ is contained in some subgraph $G_{h'}$ (except for the final path that reaches the root directly), and if $\tau$ is contained in $G_{h'}$ then its height $h(\tau)=h'$. The edges connecting $G_{h'}$ to the rest of the graph only leave vertices in $L_{h'}$, and so any path from the source to any terminal that visits $G_{h'}$ has to visit a light agent in every layer $L_j^{h'}$ of $G_{h'}$.

\subsection{The LP}

We start with the high-level overview and intuition. The LP is a natural generalization of the $1$-level LP from the previous section,
but the size of the LP now grows to $n^{O(1/\eps)}$. Consider some sub-graph $G_{h'}$. For each light agent $A$ in the last layer $L_{h'}^{h'}$ of $G_{h'}$, we have a variable $x_A$ showing to which extent $A$ is satisfied as a light agent (and equivalently, how much flow it sends to the terminals). We then write standard flow constraints on the part of the graph induced by $((\bigcup_{h'\leq h}L_{h'}^{h'}) \cup \hat{H}\cup\hat{I})$ requiring that each terminal receives one flow unit, and each agent $A \in L_h$ sends $x_A$ flow units. We also have capacity constraints requiring that at most one flow unit traverses any vertex. We will eventually perform the same LP rounding as before on this part, using the BS procedure to obtain tree decomposition and randomized rounding to select layer-$h'$ light agents in each subgraph $G_{h'}$ that will be satisfied as light agents. So this part is relatively straightforward, and similar to one from the preceding section.

After that we focus on each subgraph $G_{h'}$ separately. Our starting point is the set of randomly selected light agents to be satisfied in the last layer $h'$ of $G_{h'}$. In the rounding algorithm,  we will perform an iterative randomized rounding procedure to create a feasible integral solution that originates at the source $s$ and satisfies these agents. In iteration $j$, we will select light agents in layer $(h-j)$ that need to be satisfied. Thus we will use light agent selection as an interface between successive layers, with each iteration moving the rounding process to the next layer.
In the first iteration, for each selected agent $A$, we scale all its {\em incoming} flow from the source $s$ by a factor of $1/x_A$. In order to ensure that congestion does not grow, our LP needs to ensure that for every item $i$ in $G_{h'}$ (no matter what level it is in), the amount of flow traversing item $i$ that eventually reaches the agent $A$ is at most $x_A$. Once $A$ is chosen, we need to choose $N_A$ agents in layer $h'-1$ that will send their flow to $A$. This again is done by randomized rounding. For each $A'$ that we choose, we will scale all the relevant flow by $1/x_{A'}$, and we again need to ensure that no vertex in the graph sends more than $x_{A'}$ flow to $A'$. The need to ensure this type of capacity constraints over all layers makes the LP somewhat complex. The LP will consist of three parts. The first and the easiest part is sending the flow to the terminals directly. From this point onwards we can focus on each subgraph $G_{h'}$ separately. In the second part, we will coordinate the flow that light agents at one level of $G_{h'}$ send to light agents in the next level, ignoring the actual routing of the flow inside each level. In the third part we will perform the routing inside each level, imposing the capacity constraints on the vertices.

\paragraph{Part 1: Flow arriving directly to terminals}
Let $\lset$ be the set of light agents appearing in the last layer of each graph $G_{h'}$ (so $\lset=\bigcup_{h'}L_{h'}^{h'}$). For each light agent $A\in \lset$ we have a variable $x_A$ signifying whether or not $A$ is satisfied as a light agent, or equivalently whether or not it sends flow directly to terminals. For each terminal $t\in T$, for each light agent $A\in\lset$, let $\pset(A,t)$ be the set of all paths connecting $A$ to $t$ (notice that each such path only uses items in $\hat{I}$ and heavy agents in $\hat{H}$ and there are no additional light agents on the path. We have the following constraints:

\begin{eqnarray}
\forall t\in T&\sum_{A\in \lset}\sum_{p\in \pset(A,t)}f(p)=1&\mbox{(Each terminal receives one flow unit)}\\
\forall A\in \lset&\sum_{t\in T}\sum_{p\in \pset(A,t)}f(p)=x_A&\mbox{(Light agent $A$ sends $x_A$ flow units)}\\
\forall i\in \hat{I}&\sum_{p: i\in p}f(p)\leq 1&\mbox{(Capacity constraints)}
\end{eqnarray}

Notice that each heavy agent $A\in \hat{H}\setminus T$ has only one outgoing edge connecting it to its private item, and so it is enough to enforce the capacity constraints on the items.

This is the only part that is common to the whole graph. From now on we fix a subgraph $G_{h'}$ and describe the constraints relevant to it. For simplicity we will omit the superscript $h'$.

\paragraph{Part 2: Routing among the light agents inside $G_{h'}$}
This part will specify the {\em amount} of flow to be routed between different light agents, while we ignore the routing itself, which will be taken care of in the next part.
For clarity of exposition, we will add a layer $0$ to our graph, where $L_0=\set{s}$.

For each $j: 0\leq j\leq h'$, we define a set $\sset_j=L_{h'}\times L_{h'-1}\times\cdots\times L_j$. For each tuple $\lambda=(\ell_{h'},\ell_{h'-1},\ldots,\ell_j)\in \sset_j$ of light agents, where $\ell_k\in L_k$ for $1\leq k\leq j$, we have a variable $y(\lambda)$. The meaning of this variable in the integral solution is whether or not $\ell_j$ sends flow to $\ell_{h'}$ via a path whose only light agents are $\ell_{h'},\ell_{h'-1},\ldots,\ell_j$. Notice that $\sset_{h'}=L_{h'}$, and so we have the constraint:

\begin{equation}
\forall A\in L_{h'}\quad y(A)=x_A\quad\mbox{(Total flow of $x_A$ for each light agent $A\in L_{h'}$)}\label{total flow of xA for A}\end{equation}

Consider now some tuple $\lambda=(\ell_{h'},\ell_{h'-1},\ldots,\ell_j)\in \sset_j$. If $y(\lambda)=1$, and flow is being sent from $\ell_j$ to $\ell_{h'}$ via paths corresponding to the tuple $\lambda$, then $\ell_j$ has to receive $N_{\ell_j}$ flow units from layer $(j-1)$. For $\lambda\in \sset_j$ and $A\in L_{j-1}$, let $\lambda\circ A\in\sset_{j-1}$ be the tuple obtained by concatenating $A$ at the end of $\lambda$. So we have the constraint:

\begin{equation}
\forall 1\leq j\leq h',\forall\lambda=(\ell_{h'},\ell_{h'-1},\ldots,\ell_j)\in \sset_j\quad \sum_{A\in L_{j-1}}y(\lambda\circ A)=N_{\ell_j}\cdot y(\lambda)\label{light-agent-gets-M-flow-unts}
\end{equation}

Additionally, for $j\geq 2$, each $A\in L_{j-1}$ is only allowed to send at most $y(\lambda)$ flow units via the tuple $\lambda$ (this is similar to the capacity constraint~(\ref{capacity-constraint}) from the previous LP):

\begin{equation}
\begin{array}{lll}
\begin{array}{l}
\forall 2\leq j\leq h',\forall A\in L_{j-1}\\
\forall\lambda=(\ell_{h'},\ell_{h'-1},\ldots,\ell_j)\in \sset_j\end{array},& y(\lambda\circ A)\leq y(\lambda)&\mbox{($\lambda$-flow capacity constraints for one level)}\label{lambda-flow capacity constraints for one level}
\end{array}
\end{equation}

Finally, we need to add the more complex capacity constraints that will ensure that the randomized rounding procedure will go through. Consider some tuple $\lambda(\ell_{h'},\ell_{h'-1},\ldots,\ell_j)\in \sset_j$, and let $A \in L_k$ be any light agent in some layer $k$, where $k<j$. Then the total amount of flow that $A$ can send via all tuples whose prefix is $\lambda$ is at most $y(\lambda)$. Given $\lambda\in \sset_j$ and $A\in L_k$ with $k<j$, let $Z(\lambda, A)\sse \sset_k$ be the set of all tuples whose prefix is $\lambda$ and whose last agent is $A$. Then we have the following capacity constraints:

\begin{equation}\begin{array}{lll}
\begin{array}{l}
\forall 1\leq k< j\leq h'\\
\forall A\in L_{k},
\forall\lambda\in \sset_j\end{array},&\sum_{\lambda'\in Z(\lambda,A)}y(\lambda')\leq y(\lambda)&\mbox{($\lambda$-flow capacity constraints for multiple levels)}\label{lambda flow capacity constraints for multiple levels}\end{array}\end{equation}

Actually the set (\ref{lambda flow capacity constraints for multiple levels}) of constraints contains the constraints in (\ref{lambda-flow capacity constraints for one level}) as a special case, and we only added (\ref{lambda-flow capacity constraints for one level}) for motivating these more general constraints.
Finally to complete this part we require that each light agent sends at most one flow unit in total:

\begin{equation}\forall 1\leq j\leq h,\forall A\in L_j\quad \sum_{\stackrel{\lambda\in \sset_j:}{A\in \lambda}}y(\lambda)\leq 1\quad\mbox{(General capacity constraints for light agents)}\label{general capacity constraints for light agents}\end{equation}

\paragraph{Part 3: Routing the flow}
We focus on level $j$ of graph $G_{h'}$. Consider some tuple $\lambda=(\ell_{h'},\ldots,\ell_j,\ell_{j-1})\in \sset_{j-1}$. We will have flow $f_{\lambda}$ of type $\lambda$, and we need to send $y(\lambda)$ flow units of this type from $\ell_{j-1}$ to $\ell_j$. For any pair $\ell_{j-1}\in L_{j-1},\ell_j\in L_j$ of agents, let $\pset(\ell_{j-1},\ell_j)$ be the set of all paths connecting them (note that these paths
are completely contained inside level $j$). Then:

\begin{equation}\forall 0\leq j\leq h',\forall \lambda=(\ell_{h'},\ldots,\ell_j,\ell_{j-1})\in \sset_{j-1},\quad\sum_{p\in \pset(\ell_{j-1},\ell_j)}f_{\lambda}(p)=y(\lambda)\quad\mbox{(routing flow of each type)}\label{routing flow of each type}\end{equation}

We need to add the simple capacity constraints that the flow via any item is at most $1$:

\begin{equation}\forall 1\leq j\leq h',\forall i\in I_j\quad \sum_{p:i\in p}\sum_{\lambda\in \sset_{j-1}}f_{\lambda}(p)\leq 1\quad\mbox{(General capacity constraints)}\label{general capacity consraints}\end{equation}

Note that since each non-terminal heavy agent has exactly one out-going edge (that connects it to its private item),
the constraint above also implicitly bounds the flow through a non-terminal heavy agent
to be $1$.

And finally, we need to add capacity constraints, which are very similar to~(\ref{lambda flow capacity constraints for multiple levels}).
For a tuple $\lambda=(\ell_{h'},\ldots,\ell_j)\in \sset_j$, for each $j\leq q\leq h'$, we denote $\lambda_q=\ell_q$. Consider some tuple $\lambda(\ell_{h'},\ell_{h'-1},\ldots,\ell_j)\in \sset_j$, and let $i$ be any item in any layer $I_k$, where $k<j$. Then the total flow that $i$ can send via all tuples whose prefix is $\lambda$ is at most $y(\lambda)$. Given $\lambda\in \sset_j$ and layer $L_k$ with $k<j$, let $Z'(\lambda,k)\sse \sset_k$ be the set of all tuples whose prefix is $\lambda$. Then we have the following capacity constraints:

\begin{equation}
\begin{array}{lll}
\begin{array}{l}
\forall 1\leq k< j\leq h'\\
\forall\lambda\in \sset_j,\forall i\in I_{k}
\end{array},& \sum_{\lambda'\in Z'(\lambda,k-1)}\sum_{\stackrel{p\in\pset(\lambda'_{k-1},\lambda'_k):}{i\in p}}f_{\lambda'}(p)\leq y(\lambda)&\mbox{(multi-leveled capacity constraints)}\label{specific flow-type multi-leveled capacity constraints for items}
\end{array}\end{equation}

\paragraph{Solving the LP:}
The total number of different flow types in the LP is $O(|\sset_1|)=O(m^h)=n^{O(1/\eps)}$.
As written, the LP has exponential number of variables representing the flow-paths.
For computing a solution, we can replace it with the standard compact formulation for multi-commodity flows that specifies flow-conservation constraints, and has capacity constraints on vertices.
We can then use the standard decomposition into flow-paths to obtain a feasible solution for our LP.
Therefore the overall complexity of computing the LP solution is $n^{O(1/\eps)}$.

\subsection{The Rounding Algorithm}
The algorithm has three parts. The first part uses the BS decomposition to take care of the direct routing to the terminals. The output of the first part is the set $\pset_1$ of vertex-disjoint simple paths connecting light agents to terminals in the original graph $N(\iset,P)$. The second part is randomized rounding in each sub-graph. The third part is the ``clean-up'' phase where we get rid of almost all the congestion and create the set $\pset_2$ of paths that are used to satisfy the light agents.

\paragraph{Part 1: Routing to Terminals}
We consider the sub-graph of $N_h(\iset,P)$ induced by the set $Y=\hat{H}\cup \lset$ of agents (where $\lset$ contains the light agents in the last layer of every subgraph $G_{h'}$, $\lset=\bigcup_{h'\leq h}L_{h'}^{h'}$) and the set $\hat{I}$ of items. Recall that the items in $\hat{I}$ are heavy items for all agents in $Y$.
For each agent $A\in \lset$ we have the value $x_A$ defined by our LP-solution, while for each agent $A\in \hat{H}$ we set $x_A=0$.
 Exactly like in the first part of the rounding algorithm for Section~\ref{sec: 1-layered instances}, we can produce values $y_{A,i}$ for each $A\in Y$, $i\in \hat{I}$ satisfying the constraints (\ref{constraint 1 for BS})--(\ref{last constraint for BS}). We then again apply Theorem~\ref{thm: BS} to obtain a decomposition  of the bipartite graph $G(Y,\hat{I})$ into trees $T_1,\ldots,T_s$. Let $T_j$, $1\leq j\leq s$ be any tree in the decomposition containing more than one edge. Recall that the summation of values $x_A$ for agents $A$ in tree $T_j$ is at least $\half$. We select one of the agents $A$ of $T_j$ with probability $x_A/X$, where $X$ is the summation of values $x_{A'}$ for all agents $A'$ in $T_j$. Notice that this probability is at most $2x_A$. Once $A$ is selected, we can assign the items of tree $T_j$ to the remaining agents. Let $\lset'\sse\lset$ be the set of light agents we have thus selected. We therefore obtain an assignment of items in $\hat{I}$ to agents in $Y\setminus\lset'$, where each agent is assigned one item. This assignment of items defines an integral flow in the sub-graph of $N_h(\iset,P)$ induced by $Y\cup \hat{I}$, as follows. Every light agent in $\lset'$ sends one flow unit to its private item. If agent $A\in Y\setminus \lset'$ is assigned item $i\neq P(A)$, then $i$ sends one flow unit to $A$ and if $A\not\in T$, it sends one flow unit to $P(A)$. This flow is a collection of disjoint simple paths connecting items in $\lset'$ to the terminals. Let $\pset_1$ denote the corresponding collection of paths in $N(\iset,P)$, (where we replace copies of  agents and items back by the corresponding agents and items themselves).
 The set $\pset_1$ of paths is completely vertex disjoint: it is clear that paths in $\pset_1$ cannot share heavy agents or items. It is also impossible that a light agent $A$ has more than one path in $\pset_1$ starting at $A$: even though many copies of $A$, appearing in the last layers $L_{h'}^{h'}$ for each graph $G_{h'}$ connect to $\hat{I}$, all these copies only connect to a single copy of $h(A)$ in $\hat{I}$ and so all paths starting at copies of $A$ have to go through this vertex.

 The set $\pset_1$ of paths will not change for the rest of the algorithm and will be part of the output. We now focus on finding the set $\pset_2$ of paths that supply flow to the light agents in $\lset'$.

\paragraph{Part 2: Randomized Rounding} We focus on one subgraph $G_{h'}$, where we have a subset $L'_{h'}\sse L_{h'}$ of light agents that have been selected in Part $1$.

Given a tuple $\lambda(\ell_p,\ldots,\ell_j)\in \sset_j$, we say that flow $f_{\lambda}$ belongs to agent $\ell_p$.
 In the first iteration, for each light agent $A\in L'_{h'}$, we scale all the flow belonging to $A$ by the factor of $1/x_A$. For light agents $A\not\in L'_{h'}$, we remove all their flow from the graph.
 We define a subset of tuples $\sset'_{h'}\sse \sset_{h'}$ to be $\sset'_{h'}=L'_{h'}$.
 In this way we get a flow that is ``integral'' for layer $h'$ and fractional for remaining layers. In general, in iteration $(h'-j)$ we will get a flow that is integral for layers $j,\ldots,h'$ and fractional for the remaining layers. We will claim that the solution obtained in each iteration satisfies constraints~(\ref{total flow of xA for A})--(\ref{lambda flow capacity constraints for multiple levels}) and constraints~(\ref{routing flow of each type}),(\ref{specific flow-type multi-leveled capacity constraints for items}), while constraints (\ref{general capacity constraints for light agents}) and (\ref{general capacity consraints}) (general capacity constraints for light agents and items) are satisfied approximately, with polylogarithmic congestion.

Since each non-terminal heavy agent has exactly one out-going edge (that connects it to its private item), it is enough to bound the congestion on items and light agents.

 We now consider iteration $(h'-j)$. The input is a subset $\sset'_j\sse \sset_j$ of tuples, such that for every $\lambda\in \sset'_j$, if $\lambda=\set{\ell_{h'},\ldots,\ell_j}$, then our fractional solution routes $N_A$ flow units of types $\lambda'\in \sset_{j-1}$, where $\lambda$ is a prefix of $\lambda'$, to $A$, and for $\lambda\not\in \sset'_j$, no flow of any type $\lambda'$ where $\lambda$ is a prefix of $\lambda'$ is present. Moreover, the constraints
 (\ref{total flow of xA for A})--(\ref{lambda flow capacity constraints for multiple levels}) and constraints~(\ref{routing flow of each type}),(\ref{specific flow-type multi-leveled capacity constraints for items}) hold, while constraints (\ref{general capacity constraints for light agents}) and (\ref{general capacity consraints}) (general capacity constraints for light agents and items) are satisfied approximately, with congestion of at most $(16 h^2\cdot \log n)(1 + \frac{1}{h})^j$ for each item and each light agent.
 Let $\lambda\in \sset'_j$, and let $A\in L_j$ be the last agent in $\lambda$. Let $Z\sse \sset_{j-1}$ be the collection of tuples whose prefix is $\lambda$ and last agent is any light agent from $L_{j-1}$. Then due to Constraint~(\ref{light-agent-gets-M-flow-unts}), the summation of $y(\lambda')$ for $\lambda'\in Z$ is $N_A$. We randomly select each tuple $\lambda'\in Z$ with probability $y(\lambda')$.
Notice that by the Chernoff bound, with high probability, for each $\lambda\in\sset'_j$, at least $N_A/2$ tuples $\lambda'$ have been selected.
 If $\lambda'$ is selected, then we proceed as follows:

 \begin{itemize}

 \item Randomly sample a flow-path $p$ carrying a flow of type $\lambda'$ with probability $f_{\lambda'}(p)/y_{\lambda'}$.

 \item For all $j'<j$, scale all the values $y_{\lambda''}$ and flows of type $\lambda''$ for tuples $\lambda''\in \sset_{j'}$  such that $\lambda'$ is a prefix of $\lambda''$ by a factor of $1/y(\lambda')$.
 \end{itemize}

 If $\lambda'$ is not selected, then all flow corresponding to tuples $\lambda''$ where $\lambda'$ is the prefix of $\lambda''$ is removed from the graph. It is easy to verify that constraints (\ref{total flow of xA for A})--(\ref{lambda flow capacity constraints for multiple levels}) and constraints~(\ref{routing flow of each type}),(\ref{specific flow-type multi-leveled capacity constraints for items}) continue to hold (since both sides of such constraints get scaled by the same factor).
 We now bound the congestion on items and light agents. Consider an item $i$ at some level $k\leq j$. Due to constraint~(\ref{specific flow-type multi-leveled capacity constraints for items}), the total flow corresponding to tuples $\lambda''$ whose prefix is $\lambda'$ is at most $y(\lambda)$. Since before iteration $(h'-j)$, w.h.p.  congestion on $i$ was at most $(16 h^2\cdot \log n)(1 + \frac{1}{h})^j$, using Chernoff bound, with high probability the congestion goes up by at most a factor of $(1 + \frac{1}{h})$\footnote{For any $0 < \delta \le 2e-1$, the probability that a random variable $Z = \sum_{i=1}^{N} Z_i$, where $Z_i$'s are independent $0/1$ random variables, deviates from its expectation by more than $(1+\delta)$ is at most $e^{-(\delta^2 {\rm E}[Z])/4}$. So, if ${\rm E}[Z] \geq 16 h^2\log n $, this probability  is at most $1/n^4$ for $\delta = 1 / h$.}.
A similar argument works for bounding congestion on light agents in the fractional part, and for rounding the congestion on items and heavy agents induced by the integral paths we have selected at level $j$.

 At the end of the above procedure, we obtain a set $\pset'$ of simple paths. If $A\in L_{h'}^{h'}$ or $A$ has a simple path leaving it in $\pset'$ then with high probability it has at least $N_A/2$ paths entering it in $\pset'$. The total number of paths leaving a light agent is bounded by $O(h^2\cdot \log n)$ and similarly the total number of paths to which a heavy agent or an item can belong is bounded by $O(h^2\log n)$ with high probability.

 \paragraph{Getting Almost-Feasible Solution}
 In the last step of the algorithm we produce a set $\pset_2$ of simple paths, such that properties (\ref{property almost-feasible-first})--(\ref{property almost-feasible-last}) hold for $\pset_1$, $\pset_2$. Let $\lset^*\sse L$ be the set of light agents in the original instance from which paths in $\pset_1$ originate.

Consider the flow-paths in $\pset'$, and let $\pset''$ be the corresponding paths in the original graph $N(\iset,P)$ (where we replace copies of agents and items by their original counterparts). These flow-paths have the following properties: (i) every vertex that does not correspond to a terminal may appear in at most $\alpha' =h^4\log n$ flow-paths in $\pset''$, and (ii) for any light agent $A$, there are at most $\alpha'$ paths in $\pset''$ starting at $A$, and if there is at least one path in $\pset''\cup\pset_1$ that originates at $A$, then there must be at least $N_A/2$ paths terminating at $A$.

We now show how to convert the set $\pset''$ of paths into set $\pset_2$, such that properties~\ref{property almost-feasible-first}--\ref{property almost-feasible-last} hold for $\pset_1$, $ \pset_2$.

\begin{lemma}\label{lemma-senders-receivers}
Let $\pset$ be any collection of simple paths that all terminate at light agents, and let $\lset'$ be the subset of light agents at which these paths terminate. Moreover, assume that  for every $A\in \lset'$ there are at least $N_A/2$ paths in $\pset$ terminating at $A$, all paths in $\pset$ start at vertices in $S\cup \lset'$, and for any vertex $v$, there are at most $\beta$ paths containing $v$ as a first or an intermediate vertex. Then there is a collection $\pset^*$ of paths, such that for each $A\in \lset'$, there are $\lfloor{N_A/2\beta\rfloor}$ paths terminating at $A$, all paths in $\pset^*$ start at vertices in $\lset'\cup S$, and each vertex appears at most once as the first or an intermediate vertex of paths in $\pset^*$.
\end{lemma}

\begin{proof}
We call the light agents of $\lset'$ \emph{receivers}. A light agent $A\in \lset'$ that has at least one path of $\pset$ starting at $A$ is called a \emph{sender}.

We now build the following flow network. There is a set $\sset$ of vertices corresponding to senders and set $\rset$ of vertices corresponding to receivers (so if an agent serves both as sender and receiver, it will appear in both $\sset$ and $\rset$, and there will be two vertices representing it). Additionally, there is a
vertex for each heavy agent and for each item, and there is a source $s$ and a sink $t$. Source $s$ connects to every item in $S$ and to every sender with capacity-$1$ edges. Every receiver connects to $t$ with $\lfloor N_A/2\beta\rfloor$ parallel edges of capacity $1$ each. There is an edge from every sender to its private item, and from every heavy agent to its private item. There is an edge from item $i$ to heavy agent $A$ iff $i\in \Gamma(A)\setminus\set{P(A)}$. There is an edge from item $i$ to receiver $A$ iff $i$ is a light item for $A$. The goal is to route $\sum_{A\in \rset}\lfloor N_A/2\beta\rfloor$ flow units from $s$ to $t$.

Observe that the flow-paths in $\pset$ induce flow of value at least $\sum_{A\in \rset}N_A/2$, and violate the edge capacities by at most factor $\beta$. Therefore, if we scale this flow down by the factor of $\beta$ we will obtain a feasible flow of the desired value. From the integrality of flow, we can obtain integral flow of the same value. It is easy to see that this flow will define the desired set $\pset^*$ of paths.
\end{proof}

This finishes the algorithm for getting an almost-feasible solution.

\section{An $\tilde{O}\left(n^{\eps}\right)$-Approximation Algorithm}
\label{sec: final-algorithm}
We now show that the algorithm from Section~\ref{sec: almost feasible solutions} for obtaining an almost-feasible solution can be used as a
building block to get an $\tilde{O}\left(n^{\eps}\right)$-approximation. We use an iterative approach where each iteration $j$
begins with a canonical instance $\iset^j$, and a
subset $\lset^j$ of light agents such that each light agent $A \in \lset^j$ is {\em satisfied} using (light) items from $S(A)$.
We then generate an almost-feasible solution to the canonical instance $\iset^j$ containing agents $\aset\setminus\lset^j$. The items that satisfied agents in $\lset^j$, however, are made available for {\em re-use} in solving the iteration $j$ canonical instance. The final step is to compose the previously computed assignment to $\lset^j$ with an almost-feasible solution to the instance $\iset^j$ to generate a new set $\lset^{j+1}$ of light agents satisfied by light items, and a new canonical instance
$\iset^{j+1}$. The algorithm makes progress by reducing the number of terminals in each successive iteration, until no more
terminals remain and we have an $\tilde{O}\left(n^{\eps}\right)$-approximate assignment for all agents.

\subsection{The Algorithm}
\def\Q{\mathcal Q}
We start by defining the notion of partially satisfying a subset of light agents.
Let $\Q$ be a collection of simple paths, such that all paths in $\Q$ terminate at light agents, and
let $L'\sse L$ be any subset of light agents. We say that $\Q$ {\em $\alpha'$-satisfies $L'$} iff

\begin{itemize}
\item
the paths in $\Q$ do not share intermediate vertices,

\item
each light agent $A\in L'$ has at least $N_A/\alpha'$ paths in $\Q$ that terminate at $A$ and no path
in $\Q$ originates from $A$, and

 \item
 each light agent $A\in L\setminus L'$ has at most one path in $\Q$ originating $A$, and if such a path exists, then there are at least $N_A/\alpha'$ paths in $\Q$ that terminate at $A$.
\end{itemize}

 We now describe the algorithm in detail.
 The algorithm consists of $h$ iterations (recall that $h=8/\eps$). The input to iteration $j$ is a subset $\lset^j\sse L$ of light agents, a set $T^j\sse H$ of terminals and an assignment of private items $P^j:\aset\setminus (\lset^j\cup T^j)\rightarrow I$.  
 Additionally, in the resulting flow network $N(\iset^j,P^j)$, we have a collection $\Q^j$ of simple paths that $\alpha_j$-satisfy agents in $\lset^j$ where
$\alpha_j=2j\alpha$; here $\alpha$ is the approximation factor from Theorem~\ref{thm: almost feasible solutions}. The output of 
iteration $j$ is a valid input to iteration $(j+1)$, that is, sets $\lset^{j+1}$,$T^{j+1}$, an assignment of private items $P^{j+1}:\aset\setminus (\lset^{j+1}\cup T^{j+1})\rightarrow I$ and a collection $\Q^{j+1}$ of simple paths in $N(\iset^{j+1},P^{j+1})$ that $\alpha_{j+1}$-satisfy $\lset^{j+1}$.

The size of the set $T^j$ decreases in each iteration by a factor of at least $n^{\eps}/(32 h^2 \alpha)$, so after $h$ iterations, $|T^{h+1}|\leq |T|/(n^{\eps}/(32 h^2 \alpha))^h$. Since $h\leq \log n/\log\log n$, $n^{\eps}\geq \log^8n$ and $\alpha=O(h^4\log n)$, we have that
$32h^2\alpha = O(h^6\log n) = O(\log^7 n) = O(n^{7\eps/8})$. Thus, for large enough $n$, $(n^{\eps}/(32 h^2 \alpha))^h > (n^{\eps/8})^h = n$.

Therefore, $|T^{h+1}|<1$, that is,  $T^{h+1}$ is empty and  $P^{h+1}$ assigns a private item to each agent in $\aset\setminus \lset$. Furthermore, $N(\iset^{h+1},P^{h+1})$ contains
a collection $\Q^{h+1}$ of paths that $\alpha_{h+1}$-satisfy the agents in $\lset^{h+1}$ -- the only agents without private items. The flow-paths in $Q^{h+1}$ together with the assignment $P^{h+1}$ of private items then define an $\alpha_{h+1}=2(h+1)\alpha$-approximate solution.\\

In the input to the first iteration, $\lset^1=\emptyset$, $\Q^1=\emptyset$. Each light agent $A\in L$ is assigned its heavy item as private item, $P^1(A)=h(A)$, and the assignment of private items to heavy agents is performed by calculating a maximum matching between the set of heavy agents and the remaining items. This is as in Section \ref{sec:private}.

Iteration $j$ (for $j = 1$ to $h$) is performed as follows. We construct a canonical instance $\iset^j$ that is identical to $\iset$ except that we remove the light agents in $\lset^j$ from this instance. Let $\Nj=N(\iset^j,P^j)$ be the corresponding flow network. Note that we do not remove any items from the instance.
Thus, the integral optimum for this instance cannot decrease, and in particular the value of the optimal solution is at least $(h+1)\cdot N_A$ (recall that we have scaled the values $N_A$ down by factor $(h+1)$). Therefore, there exists an $h$-layered solution of value $N_A$, implying the LP described in Section ~\ref{sec: almost-feasible-solutions} must be feasible for this instance as well. Conversely, if the LP is not valid for $\Nj$, we say that the problem is infeasible.

We now apply the algorithm from Section~\ref{sec: almost-feasible-solutions} to $\Nj$, obtaining the two sets $\pset_1$, $\pset_2$ of simple paths, satisfying the properties~\ref{property almost-feasible-first}--\ref{property almost-feasible-last} as in Theorem \ref{thm: almost feasible solutions}.

Let $\lset'$ be the subset of light agents $A$ for which there is a path in either $\pset_1$ or $\pset_2$ originating from $A$. Recall that for any such $A$, there are at least $N_A/\alpha$ paths in $\pset_2$ terminating at $A$. Let $\lset''$ be the set of light agents $A$ such that either $A\in \lset^j$, or there is a path in $\Q^j$ originating from $A$. Recall that there are at least $N_A/\alpha_j$ paths terminating at $A$ in $\Q^j$.

Our first step is to construct a set $\Q^*$ with the following property.

\begin{claim}
There exists a set of internally-disjoint simple paths $\Q^*$ such that each agent $A$ in $\lset'\cup\lset''$ has at least $\lfloor N_A/(\alpha_j+\alpha)\rfloor$ paths terminating at $A$. Moreover, only light agents in $(\lset'\cup \lset'')\setminus \lset^j$ have paths in $\Q^*$ originating from them, with at most one path originating from any agent.
\end{claim}
\begin{proof}
This is done similarly to Lemma~\ref{lemma-senders-receivers}, where light agents in $\lset'\cup\lset''$ serve as receivers and light agents in $(\lset'\cup \lset'')\setminus \lset^j$ are the senders. Each receiver $A$ is connected to the sink $t$ with $\lfloor N_A/(\alpha_j+\alpha)\rfloor$ edges of capacity $1$.
A source connects to all senders and items in $S$ with edge of capacity $1$.
Consider the flows defined by paths in $\pset_2$ and $\Q^j$. We send $\alpha/(\alpha_j+\alpha)$ flow units along each path in $\pset_2$ and $\alpha_j/(\alpha_j+\alpha)$ flow units along each path in $\Q^j$. The resulting flow causes congestion of at most $1$ on the edges, and each receiver $A$ gets at least $N_A/(\alpha_j+\alpha)$ flow units. From the integrality of flow, there is a collection $\Q^*$ of desired paths.
\end{proof}

Our next step is to resolve the conflicts between paths in $\Q^*$ and $\pset_1$ when such paths share vertices. We will use Lemma~\ref{spider-lemma} below to do so. The idea is to re-route the paths in $\pset_1$ to get $\pset'_1$, so that each such path only interferes with at most one path in $\Q^*$. We will then remove from $\Q^*$ the paths that share vertices with the re-routed paths and argue that we still make progress.

\paragraph{A Path Rerouting Lemma:}
For a directed path $p$ starting at some vertex is $v$, we say that path $p'$ is a \emph{prefix} of $p$ iff $p'$ is a sub-path of $p$ containing $v$.
We will use the following lemma whose proof follows from the Spider Decomposition Theorem of~\cite{CK08} and appears in Appendix.

\begin{lemma}\label{spider-lemma}
Let $\pset,\Q$ be two collections of directed paths, such that all paths in $\pset$ are completely vertex disjoint and so are all paths in $\Q$.
We can define, for each path $p\in\pset\cup \Q$, a prefix $\gamma(p)$, such that if $\cset$ is a connected component in the graph $G_\gamma$ defined by the union of prefixes $\set{\gamma(p)\mid p\in \pset \cup\Q}$, then

\begin{itemize}
\item either $\cset$ only contains vertices belonging to a single prefix $\gamma(p)$ and $\gamma(p)=p$, or

\item $\cset$ contains vertices of exactly two prefixes $\gamma(p)$ and $\gamma(q)$, where $p\in \pset$, $q\in \Q$, and the two prefixes have exactly one vertex in common, which is the last vertex of both $\gamma(p)$ and $\gamma(q)$.
\end{itemize}
\end{lemma}

\paragraph{Rerouting Paths in $\pset_1$ and $\Q^*$:}
Recall that the paths in $\pset_1$ are completely vertex-disjoint. Paths in $\Q^*$ may share endpoints, but for each agent $A$ there is at most one path in $\pset^*$ that originates from $A$. In order to apply Lemma~\ref{spider-lemma} however we need to ensure that paths in $\Q^*$ are completely vertex-disjoint. For each path $q\in \Q^*$, if $A$ is the last vertex on $q$, we introduce a new dummy vertex $v(q,A)$ that replaces $A$ on path $q$. Let $\Q^{**}$ be the resulting set of paths. We also transform set $\pset_1$ as follows.
Given a directed path $p$, we denote by $\overline{p}$ the path obtained by reversing the direction of all edges of $p$. We define $\pset^*_1=\set{\overline{p}\mid p\in \pset_1}$. In the new set $\pset^*_1$, the paths originate at the terminals and terminate at light agents.  We now apply Lemma~\ref{spider-lemma} to $\pset^*_1$ and $\Q^{**}$ and obtain prefix $\gamma(p)$ for each $p\in \pset^*_1\cup \Q^{**}$.

For each $p\in \pset^*_1$, we construct a new path $p'$ that will be used to re-route the flow to the terminal of $p$. 
Consider the connected component $\cset$ in the graph $G_{\gamma}$ induced by the prefixes, to which $\gamma(p)$ belongs. If $\cset$ only contains vertices of $\gamma(p)$, then $p=\gamma(p)$ and we set $p'=\overline{p}$. Otherwise, $\cset$ contains vertices of $p$ and another path $q\in \Q^{**}$.
 Consider the vertex $v$ that is common to $\gamma(p)$ and $\gamma(q)$.
 If $v$ is a light agent, (notice that in this case since $p$ and $q$ are simple, $\gamma(p)=p$ and $\gamma(q)=q$ must hold), then we set $p'=\overline{p}$. Otherwise, we set $p'$ to be the concatenation of $\gamma(q)$ and $\overline{\gamma(p)}$. In either case path $q$ is removed from set $\Q^*$, and we say that the unique terminal $t$ that lies on path $p$ is \emph{responsible} for the removal of $q$. We note here that a terminal will be responsible for the removal of {\em at most} one path $q$.

 Now observe that the first vertex of $q$ has to be a light agent $A$. For if it is an item in $S$, then the path $\gamma(q)$ followed by $\overline{\gamma(p)}$ is a simple path from an item in $S$ to a terminal, which is impossible.
Therefore, in both cases above, the new path $p'$ has one end point a terminal and the other is a light agent $A$.
This implies $\Q^*$ originally contained at least $\lfloor N_A/(\alpha_j+\alpha) \rfloor$ paths terminating at $A$.

We denote by $\pset'_1$ the resulting set containing paths $p'$ for all $p\in \pset_1$, and by $\Q_2$ the set of remaining paths in $\Q^{*}$. The set ($\pset'_1\cup \Q_2$)  ``almost'' has all the desired properties of the final solution: all paths are internally vertex-disjoint; {\em each terminal} has exactly one path entering it and each light agent has at most one path leaving it. Moreover, if light agent $A$ has a path leaving it, then {\em originally}, in $\Q^*$, there were at least $\lfloor N_A/(\alpha_j+\alpha)\rfloor$ paths terminating at $A$. A potential problem is that it is possible that we have removed many of such paths
on moving to $\Q_2$. We now take care of that.

\paragraph{Bad Light Agents:}
We call a light agent $A$ is {\em bad} iff there is a path originating at $A$ in $\pset'_1\cup \Q_2$ but there are less than $N_A/(\alpha_j+2\alpha)=N_A/\alpha_{j+1}$ paths terminating at $A$.

We start with $T^{j+1}=\emptyset$.
While there exists a bad light agent $A$:

\begin{itemize}

\item Remove all paths entering $A$ from $\Q_2$.

\item If $A\not\in \lset^j$, then remove the unique path $p$ leaving $A$ from $\pset'_1$ or $\Q_2$, and say that $A$ is \emph{responsible} for this path. If $p\in \pset'_1$ and $t$ is the terminal lying on $p$, then we add $t$ to $T^{j+1}$ and say that $A$ is responsible for $t$.

\item If $A\in \lset^j$, consider the item $i=h(A)$. If there is a heavy agent $A'$ for which $i$ is a private item, we add $A'$ to $T^{j+1}$ (where it becomes a terminal). In either case, item $i$ becomes the private item for $A$. We remove $A$ from $\lset^j$.
If there is any path $p$ containing $i$ in $\pset'_1\cup \Q_2$, then we remove $p$ from $\pset'_1$ or $\Q_2$ and say that $A$ is responsible for $p$. If $p\in \pset'_1$ and $t$ is a terminal lying on $p$, then we add $t$ to $T^{j+1}$ and say that $A$ is responsible for $t$.
\end{itemize}

It is easy to see that a bad light agent can only be responsible for at most one path in $\pset'_1\cup \Q_2$, and at most two terminals in $T^{j+1}$. Notice that once we take care of a bad light agent $A$, this could result in another agent $A'$ becoming a bad light agent. We repeat this process until no bad light agents remain. We show below that the size of $T^{j+1}$ is small, but first we show how to produce the input to the next iteration (which is the output of the current iteration).

\paragraph{Input to Iteration $(j+1)$:}
We start with $\lset^{j+1}$ containing all the remaining good agents in $\lset^j$.
Consider now the sets $\pset'_1\cup \Q_2$ of paths. Let $p\in \pset'_1$, and let $A$ be the first vertex and $t\in T^j$ be the last vertex on $p$. We then add $A$ to $\lset^{j+1}$ and re-assign private items that lie on path $p$ as follows. If $A'$ is the agent lying immediately after item $i$ on path $p$ then $i$ becomes a private item for $A'$. The assignment of private items of agents not lying in any path remains the same.
Let $P^{j+1}$ be the resulting assignment of private items. Note that the only agents with no private items assigned are agents of $\lset^{j+1}\cup T^{j+1}$. We set $\Q^{j+1}=\Q_2$. Since no light agent  in $\lset^{j+1}$ is bad, set $\Q^{j+1}$ ensures that every agent in $\lset^{j+1}$ is $\alpha_{j+1}$-satisfied. Therefore we have produced a feasible input to iteration $(j+1)$.

\paragraph{Bounding the size of $T^{j+1}$:}
Finally we need to bound the size of $T^{j+1}$, the set of terminals in iteration $(j+1)$.

\begin{lemma}\label{lemma: bound on size of Tj}
$|T^{j+1}|\leq   \left( \frac{32 h^2 \alpha}{n^{\eps}} \right) |T^j|$.
\end{lemma}

\begin{proof}
We may assume that  $n^{\eps} \geq 16h^2 \alpha$.
Recall that each bad light agent is responsible for at most two terminals in $T^{j+1}$. Therefore, it is enough to prove that the number of bad light agents in iteration $j$ is at most $\left( \frac{16h^2\alpha}{n^{\eps}} \right)|T^j|$. We build a graph $G_B$ whose vertices are bad light agents and the terminals in $T^j$. Consider now some bad light agent $A$. Originally there were at least $\left\lfloor \frac{N_A}{(\alpha_j+\alpha)} \right\rfloor \geq \frac{n^{\eps}}{(2j+1)\alpha}-1$ paths entering $A$ in $\Q^{*}$. Since $A$ is a bad light agent, eventually less than $n^{\eps}/((2j+2)\alpha)$ paths remained. Therefore, at least $\frac{n^{\eps}}{(2j+1)(2j+2)\alpha-1}\geq \frac{n^{\eps}}{8h^2\alpha}$ paths have been removed from $\Q^{*}$. If $q$ is such a path, and $A'$ is responsible for $q$ then we add an edge from $A$ to $A'$ in graph $G_B$ (observe that $A'$ can be another bad light agent or a terminal).

Since any bad light agent or terminal is responsible for the removal of at most one path, the in-degree of every vertex is at most $1$. Moreover by the discussion above we see the out-degree of every bad light agent is at least  $\beta=\frac{n^{\eps}}{8h^2\alpha} \ge 2$. Note that the out-degree of terminals in $T_j$ is $0$.
Let $n_B$ be the number of bad light agents  in $G_B$. Since the sum of in-degrees equal the sum of out-degrees, we have
$$ n_B + |T_j| \ge \beta n_B$$
implying the number of bad light agents is bounded by $|T_j|/(\beta - 1) \le 2|T_j|/\beta$ since $\beta \ge 2$.
\end{proof}

\subsection{Approximation Factor and Running Time}

Let $\alpha^*$ be the approximation factor that we achieve for the canonical instance. The final approximation factor is $\max\set{O(n^{\epsilon}\log n),O(\alpha^*\log n)}$.
We now bound $\alpha^*$ in terms of $\alpha= O(h^4\log n)$, the approximation factor from Theorem~\ref{thm: almost feasible solutions}.

We lose a factor of $(h+1) \le 2h$ when converting the optimal solution to an $h$-layered forest.
The algorithm in the final section assigns $N_A/(2h\alpha)$ items to each light agent $A$. So overall $\alpha^* =O(h^2\alpha)=O(h^6\log n)$). So we get a $\max\set{O(n^{\eps}\log n),O(h^6\log n)}$-approximation. When $\eps$ is chosen
to be $(8 \log \log n)/\log n$, we get an $O(\log^9 n)$-approximation algorithm.

The bottleneck in the running time of our algorithm is solving the linear program. As noted earlier,
 the time taken for solving the LP is $n^{O(1/\eps)}$. For our choice of $\eps$ above, we get an overall
 running time of $n^{O(\log n / \log \log n)}$.

\subsection{A Quasi-Polynomial Time $O(m^\eps)$-Approximation Algorithm}
In this section we show how to use the quasi-polynomial time $O(\log^9 n)$-factor algorithm to obtain an $O(m^\eps)$ algorithm
for any fixed $\epsilon > 0$.
Before that we make the following claim.

\begin{claim}
There exists an $(\log n)^{O(m \log n)}$-time $O(1)$-approximation to \MMA.
\end{claim}
\begin{proof}
From Section \ref{sec:polybound} we can assume all the utilities $u_{A,i}$ to be between $1$ and $2n$. By losing another constant
factor in the approximation, we round down all the utilities to the nearest power of $2$. Thus there are $O(\log n)$ distinct values of utilities.
We assume we are given an instance like this.

For every agent $A$, we let $v_j(A)$ be the number of items with $u_{A,i} = 2^j$, for $j=1$ to $s =\lfloor \log 2n \rfloor$. Thus, the optimum solution
corresponds to $m$ vectors $v(A) := (v_1(A),\ldots,v_s(A))$. At the cost of losing another factor of $2$, we can further assume that each of the $v_i(A)$
is a power of $2$. Therefore for every agent there are at most $(\log n)^s$ possible vectors $v(A)$, and one of them corresponds to the optimal solution.

We now show how given $v(A)$ for every agent, we can check if there is a feasible assignment of the items respecting $v(A)$, that is, each agent $A$ gets $v_j(A)$ items of utility $u_{A,i} = 2^j$. Construct a bipartite graph $G(U,V,E)$ where $U$ contains $s$ copies of each agent $A$: $A(1),\ldots,A(s)$, and the vertex set
$V$ corresponds to the set of items. An edge goes from $A(j)$ to item $i$ iff $u_{A,i} = 2^j$. The problem of checking if whether the vector $v(A)$ for every $A$ can be realized is equivalent to testing if there exists a matching from $U$ to $V$ such that every vertex $A(j)$ in $U$ has exactly $v_j(A)$ edges incident on it and every vertex in $V$ has one edge incident on it. This can be done in polynomial time.

Thus in time $(\log n)^{O(m \log n)}$ (over all the choices of vectors of all agents), we can get an $O(1)$ approximation to \MMA.
\end{proof}

Using the above claim we can get the following.
\begin{theorem}
For any constant $\eps > 0$, there exists a quasi-polynomial time algorithm which returns a $O(m^\eps)$-approximation to \MMA.
\end{theorem}

If $m < \log^{9/\eps} n$, then by the claim above we can get a $O(1)$ approximation in $(\log n)^{O(m \log n)}$ time which is quasi-polynomial if $\eps$ is a constant. If $m \ge \log^{9/\eps} n$, then our main result gives a quasi-polynomial time $O(\log^9 n) = O(m^\eps)$-factor algorithm.

\section{The $2$-Restricted \MMA problem}
In this section we focus on the restricted version of \MMA, where each item $i$ is wanted by at most $2$ agents.

\begin{definition}
A \MMA problem instance is {\em $2$-restricted} if for each item $i$ there exist at most two agents with $u_{A,i} > 0$. For item $i$, we denote these two agents
as $A_i$ and $B_i$. Note that $A_i=B_i$, if the item is wanted by only one agent. The $2$-restricted \MMA instance is {\em uniform} if for every item $i$, $u_{A_i,i} = u_{B_i,i}$.
\end{definition}

Given a $2$-restricted \MMA instance $\I(\A,I)$, we construct an undirected graph $G(\I)$ as follows. The set of vertices of $G(\I)$ is the set $\A$ of agents, and for every item $i\in I$,
we have an edge $i=(A_i,B_i)$. Note that $G(\I)$ can have parallel edges and self-loops. An edge corresponding to an item $i$ has two weights associated with it, one weight for each endpoint: $w_{A_i,i} = u_{A_i,i}$ and $w_{B_i,i} = u_{B_i,i}$. The interpretation of these weights is as follows: if the edge is oriented from $B_i$ to $A_i$, then its weight is $w_{A_i,i}$, and if it is oriented towards $B_i$ then its weight is $w_{B_i,i}$. Such a weighted graph will be called a {\em non-uniformly} weighted graph.
If $w_{A_i,i}=w_{B_i,i}$ for all edges $i$, the graph is called {\em uniformly weighted}.
The $2$-restricted \MMA problem on instance $\I(\A,I)$
is equivalent to the following orientation problem on $G(\I)$.\\
\\
\noindent
\def\O{\mathcal O}
{\bf Non-uniform Graph Balancing}:
Consider a graph $G(V,E)$ that can have self-loops and parallel edges, where for each edge $e=(u,v)$ we are given
two weights $w_{u,e}$ and $w_{v,e}$.
Given an orientation $\O$ of edges, for an edge $(u,v)$ we denote $(u\stackrel{\O}{\rightarrow} v)$ if the edge is oriented towards $v$ and $(v\stackrel{\O}{\rightarrow}u)$ when it is oriented towards $u$. The weighted in-degree of a vertex $v$ is
 $\sum_{\stackrel{e=(u,v)\in E: }{(u\stackrel{\O}{\rightarrow} v)}} w_{v,e}$. The goal is to find an orientation of the edges
such that the minimum weighted in-degree of a vertex is maximized.\\
 \noindent
In the uniform version of the graph balancing problem, $w_{u,e}=w_{v,e}$ for each edge $e=(u,v)$.

It is easy to see that the non-uniform (uniform) graph-balancing problem is equivalent to the non-uniform (uniform) $2$-restricted \MMA.
In this section we give a $2$-approximation for the $2$-restricted non-uniform \MMA problem. We also show that even the uniform version of the problem is NP-hard to approximate to within a factor better than $2$.
\def\C{\mathcal C}

\subsection{Approximation Algorithm for Non-Uniform Graph Balancing}
Let $\iset=(\aset,I)$ be the input instance of the $2$-restricted \MMA problem, and let $G=(\aset,I)$ be the corresponding instance of the non-uniform graph balancing problem. We start by guessing the value $\opt$ of the optimal solution via binary search.

For an agent $A\in \aset$, let $\delta(A)$ denote the set of adjacent edges in $G$, that is, $\delta(A) := \{i: u_{A,i} > 0\}$
Given a parameter $M>0$, let $\C(A)$ define a set of feasible configurations for $A$ with respect to $M$, that is: $\C(A) := \{S \subseteq \delta(A): \sum_{i\in S} u_{A,i} \ge M\}$.
For any $M$, we define the following system $LP(M)$ of linear inequalities.
We have a variable $z_{A,S}$ for every agent $A$ and every $S\in \C(A)$, that indicates whether or not $S$ is chosen for $A$.
This system of inequalities is equivalent to the configuration LP of Bansal and Sviridenko \cite{BS}.

\begin{equation*}
\begin{array}{lll}
LP(M):&&\label{eq:CONFLP} ~~~~~~~~~~~~~~~~~~~~~~~~~~~~~~~~~~~~~~~~~~~~~~~~~~~~~ \\
&\forall A\in \A: & \sum_{S\in \C(A)} z_{A,S} = 1 \\
&\forall i\in I: & \sum_{S\in \C(A_i): i\in S} z_{A_i,S} +  \sum_{S\in \C(B_i): i\in S} z_{B_i,S} = 1  \\
&\forall A\in \A, \forall S\in \C(A): & z_{A,S} \ge 0
\end{array}
\end{equation*}

Notice that $LP(M)$ has a feasible solution for $M=\opt$, the optimal solution value for instance $\iset$.
It is well known (see for instance \cite{BS,AS}) that if $LP(M^*)$ is feasible, then a solution with value $(1-\eps)M^*$ can be found in polynomial time, for any $\eps>0$. We will therefore assume that we are given a feasible solution $z$ for $LP(M)$ with $M \ge (1-\eps)\opt$.

Given agent $A$ and item $i\in \delta(A)$, we say that item $i$ is integrally allocated to $A$ iff $\sum_{\stackrel{S\in \C(A):}{i\in S}}z_{A,S}=1$. If item $i$ is not allocated integrally to $A_i$ or $B_i$ then we say that it is allocated fractionally. Let $I'$ denote the set of items allocated fractionally, and for every agent $A\in \aset$, let $I(A)$ denote the set of items integrally allocated to $A$. We set $M_A=M-\sum_{i\in I(A)}u_{A,i}$.

We now focus on the sub-graph $H$ of $G(\iset)$ induced by the edges corresponding to items in $I'$. We remove from $H$ all isolated vertices. Observe that $H$ does not contain self-loops but may contain parallel edges. Moreover, for each agent $A$ in $H$, $M_A>0$.
It now suffices to allocate the items of $I'$ to the agents of $H$, such that every agent $A$ gets a utility of at least $M_A/2$ -- this will give
a $(2+\eps)$-approximation algorithm. For each agent $A$, let $\delta'(A)$ denote the set of edges adjacent to $A$ in $H$.

\indent
We begin with the following observation about $H$.
\begin{claim}\label{claim: utility-bound}
Fix an agent $A$. Let $i^*$ be an item in $\delta'(A)$ with maximum value $u_{A,i}$. Then
$$\sum_{i\in \delta'(A)\setminus \set{i^*}} u_{A,i} \ge M_A$$
That is, the total utility of the fractional items having positive utility for agent $A$ is at least $(M_A + u_{A,i^*})$.
\end{claim}
\begin{proof}
\def\F{\mathcal F}
Since item $i^*$ is fractionally allocated, there is a configuration $S\in \C(A)$ with $z_{S,A}>0$ such that $i^*\not\in S$. Clearly, $\sum_{i\in S\cap I'} u_{A,i} \geq M_A$.
\end{proof}

The above claim implies that for every vertex $A\in V(H)$, the
non-uniform weighted degree of $A$ is at least $(M_A + \max_{i\in\delta'(A)} \set{u_{A,i}})$. We now prove a theorem about weighted
graph orientations that will complete the proof.

Given an arbitrary graph $G=(V,E)$, for each vertex $v\in V$, we denote by $\Delta(v)$ the set of edges incident on $v$. Given an orientation $\oset$ of edges, we denote by $\Delta^-_{\oset}(v)$ and $\Delta^+_{\oset}(v)$ the set of incoming and outgoing edges for $v$, respectively.

\def\O{\mathcal O}
\begin{theorem}\label{thm:weo}
Given a non-uniformly weighted undirected graph $G(V,E)$ with weights $w_{u,e}$ and $w_{v,e}$ for every edge $e=(u,v)\in E$, there exists an orientation $\O$ such that in the resulting digraph, for every vertex $v$:
$$\sum_{e\in\inedges{v}} w_{v,e} \ge \frac{\sum_{e\in \Delta(v)} w_{v,e} - \max_{e\in \Delta(v)}\set{w_{v,e}}}{2}$$
\end{theorem}

We apply Theorem \ref{thm:weo} to graph $H$. The resulting orientation of edges implies an assignment of items in $\iset'$ to agents of $H$. It is easy to see from Claim~\ref{claim: utility-bound} that the total utility of items assigned to any agent $A$ of $H$ is at least $M_A/2$, and thus together with the assignment of the integral items we obtain a factor $(2+\eps)$-approximation for the $2$-restricted \MMA problem. We now turn to prove Theorem~\ref{thm:weo}.

\begin{proof}
The proof is by induction on the number of edges of the graph. If the graph only contains one edge, the theorem is clearly true.
Consider now the case that there is some vertex $u\in V$ with $|\Delta(u)|=1$. Let $e=(u,v)$ be the unique edge incident on $u$. We can direct $e$ towards $v$, and by induction there is a good orientation of the remaining edges in the graph. Therefore we assume that every vertex in the graph has at least two edges incident on it. For each vertex $v\in V$, we denote by $e_1(v)\in \Delta(v)$ the edge $e$ with maximum value of $w_{v,e}$ and by $e_2(v)$ the edge with second largest such value. Notice that $e_1(v)$ and $e_2(v)$ are both well defined, and it is possible that they are parallel edges.
We now need the following claim.
\begin{claim}\label{claim:cycle}
We can find a directed cycle $C=(v_1,v_2,\ldots,v_k=v_1)$, where for each $j: 1\leq j\leq k-1$, $h_j=(v_j,v_{j+1})\in E$, and either:
 \renewcommand{\theenumi}{\arabic{enumi}}
 \begin{enumerate}
 \item $w_{v_j,h_{j-1}} \ge w_{v_j,h_j}$, or
\item $h_{j}=e_1(v)$ and $h_{j-1}=e_2(v)$
\end{enumerate}

\end{claim}
We first show that the above claim finishes the proof of the theorem.
Let $C$ be the directed cycle from the above claim. We remove the edges of $C$ from the graph and find the orientation of the remaining edges by induction. We then return the edges of $C$ to the graph, with edge $h_j=(v_j,v_{j+1})$ oriented towards $v_{j+1}$, for all $j$.
By rearranging the inequality that we need to prove for each vertex $v$ we obtain the following expression:

$$\sum_{e\in\inedges{v}} w_{v,e} +  \max_{e\in \Delta(v)}\set{ w_{v,e}}\ge \sum_{e\in \outedges{v}} w_{v,e}$$

Consider some vertex $v\in V$. If $v$ does not lie on the cycle $C$, then by induction hypothesis the inequality holds for $v$. Assume now that $v=v_j\in C$. In the orientation of edges of $E\setminus C$ the above inequality holds by induction hypothesis. We need to consider two cases. If $w_{v_j,h_{j-1}} \ge w_{v_j,h_j}$, then since $h_{j-1}$ is added to $\inedges{v}$ and $h_j$ is added to $\outedges{v}$, the inequality continues to hold.

Assume now that $h_{j}=e_1(v)$ and $h_{j-1}=e_2(v)$, and let $e_3(v)$ be the edge with third largest value of $w_{v,e}$. Then the RHS increases by $w_{v,h_j}$, while the LHS increases by $w_{v,h_{j-1}}+w_{v,h_j}-w_{v,e_3(v)}$. Since $w_{v,j_{h-1}}\geq w_{v,e_3(v)}$ the inequality continues to hold.

\begin{proofof}{Claim \ref{claim:cycle}}
We start with an arbitrary vertex $v_1\in V$ and add $v_1$ to $C$. In iteration $j$ we add one new vertex $v_j$ to $C$, until we add a vertex $u$ that already appears in $C$. Assume that the first appearance of $u$ on $C$ is $u=v_r$. We then remove vertices $v_1,\ldots,v_{r-1}$ from $C$ and reverse the orientation of $C$ to produce the final output (so vertices appear in reverse order to that in which they were added to $C$).

In the first iteration, $C=\set{v_1}$. Let $e_1(v_1)=(v_1,u)$. We then add $u$ as $v_2$ to $C$. In general, in iteration $j$, consider vertex $v_j$ and the edge $h_{j-1}=(v_{j-1},v_j)$. If $h_{j-1}\neq e_1(v_j)$, and $e_1(v_j)=(v_j,u)$, then we add $u$ as $v_{j+1}$ to $C$. Otherwise, let $e_2(v_j)=(v_j,u')$. We then add $u'$ as $v_{j+1}$ to $C$.

Let $v_{t+r}$ be the last vertex we add to the cycle, and assume that $v_{t+r}=v_r$. Consider the cycle $C'=(v_r,v_{r+1},\ldots,v_{t+r}=v_r)$ (recall that we will reverse the ordering of vertices in $C'$ in the final solution, the current ordering reflects the order in which vertices have been added to $C$). We denote $g_j=(v_j,v_{j+1})$.
Consider now some vertex $v_j\in C$. Assume first that $j=r$. Then two cases are possible. If $g_r=e_1(v_r)$, then clearly $w_{v_r,g_r}\geq w_{v_r,g_{r+t-1}}$ and condition (1) will hold in the reversed cycle. Assume now that $g_r=e_2(v_r)$. Then the edge $e'=(v_{r-1},v_r)$ that originally belonged to $C$ is $e_1(r)$, and so $w_{v_r,g_r}\geq w_{v_r,g_{r+t-1}}$ still holds.

Assume now that $j\neq r$. If $g_j=e_1(v_j)$ then clearly $w_{v_j,g_j}\geq w_{v_j,g_{j-1}}$ and condition (1) holds. Otherwise it must be the case that $g_{j-1}=e_1(v_j)$ and $g_j=e_2(v_j)$ and so condition (2) holds.
\end{proofof}
\end{proof}

\subsection{Hardness of Approximation for Uniform Graph Balancing}
We show that the uniform graph balancing is NP-hard to approximate up to a factor $2-\delta$ for any $\delta>0$. This result implies the same hardness of approximation for the uniform $2$-restricted \MMA, since the two problems are equivalent.

\begin{theorem}
Uniform graph balancing is NP-hard to approximate to within a factor $2-\delta$, for any $\delta > 0$.
\end{theorem}
\begin{proof}
This proof is similar to an NP-hardness of the min-max version of graph balancing due to Ebenlendr et.al. \cite{EKS}.
We reduce from the following variant of $3$-SAT. The input is a 3CNF formula $\phi$, where each clause has $3$ variables, and each {\bf literal} appears in at most $2$ clauses and at least $1$ clause. This version is NP-hard~\cite{PY91}
\noindent

We construct a graph $G=(V,E)$ whose vertices correspond to the literals and clauses in the formula $\phi$. We define edges of $G$ and the weights associated with them as follows. For every variable $x$ we have two vertices $x,\notx$ representing the two corresponding literals, with an edge $(x,\notx)$ of weight $1$. We refer to edges of this type as \emph{variable edges}. Consider now some clause $C=(\ell_1\bor\ell_2\bor\ell_3)$. We have a vertex $C$ and three \emph{clause} edges $(C,\ell_1)$, $(C,\ell_2)$, $(C,\ell_3)$ associated with it. All clause edges have weight $\half$. Additionally, we have a self-loop of weight $\half$ for each clause $C$, and for each literal $\ell$ appearing in exactly one clause.

\noindent
{\bf YES case}: Assume that the formula is satisfiable. We show that the optimum value of the graph balancing instance is $1$. Consider any variable $x$. If the optimal assignment gives value $T$ for $x$, then we orient the corresponding variable edge towards $x$; otherwise it is oriented towards $\notx$. For each literal $\ell$ set to $F$, orient all its adjacent clause edges towards it. Together with the self-loop, there are $2$ edges oriented towards $\ell$, each of which has weight $\half$. Finally, consider clause $C$. Since this clause is satisfied by the assignment, at least one of its variable has value $T$, and we can orient the corresponding edge towards $C$. Together with the self-loop, $C$ has $2$ edges oriented towards it of weight $\half$ each.

\noindent
{\bf NO case}: Suppose there is an allocation such that every vertex has weighted in-degree strictly more than $1/2$. This implies that every vertex has weighted in-degree at least $1$.
Consider the orientation of the variable edges. This orientation defines an assignment to the variables: if $(x,\notx)$ is oriented towards $x$, the assignment is $T$, otherwise it is $F$. Consider a literal $\ell$ that does not have the variable edge oriented towards it. Then it must have the $2$ remaining edges incident on it
 oriented towards it (since they have weight $1/2$ each). Now consider a clause vertex $C$. Since its weighted in-degree is at least $1$,
it must have a clause edge oriented towards it. The corresponding literal then is assigned the value $T$ and therefore $C$ is satisfied. But we know at least one clause is not satisfied by the above assignment. This proves the theorem.
\end{proof}


\appendix

\section{Proof of Lemma~\ref{spider-lemma}}
We use the following definitions from~\cite{CK08}.

\begin{definition} [Canonical Spider]
Let $\tset$ be any collection of simple paths, such that each path
$p\in \tset$ has a distinguished endpoint $t(p)$, and the other
endpoint is denoted by $v(p)$. We say that paths in $\tset$ form a
\emph{canonical spider} iff $|\tset|>1$ and there is a vertex $v$,
such that for all $p\in \tset$, $v(p)=v$. Moreover, the only
vertex that appears on more than one path of $\tset$ is $v$. We
refer to $v$ as the \emph{head} of the spider, and the paths of
$\tset$ are called the \emph{legs} of the spider.
\end{definition}

\begin{definition} [Canonical Cycle]
Let $\tset=\set{g_1,\ldots,g_h}$ be any collection of simple
paths, where each path $g_i$ has a distinguished endpoint $t(g_i)$
that does not appear on any other path in $\tset$, and the other
endpoint is denoted by $v(g_i)$. We say that paths of $\tset$ form
a \emph{canonical cycle}, iff
(a) $h$ is an odd integer,
(b) for every path $g_i$, $1\leq i\leq h$, there is a vertex
$v'(g_i)$ such that $v'(g_i)=v(g_{i-1})$ (here we use the
convention that $g_0=g_h$), and
(c) no vertex of $g_i$ appears on any other path of $\tset$,
except for $v'(g_i)$ that belongs to $g_{i-1}$ only and $v(g_i)$
that belongs to $g_{i+1}$ only.
\end{definition}

\begin{theorem}(Theorem 4 in~\cite{CK08})
\label{thm: CK-spider decomposition}
Given any collection $\pset$ of paths, where
every path $f\in \pset$ has a distinguished endpoint $t(f)$ that
does not appear on any other path of $\pset$, we can find, in
polynomial time, for each path $f\in \pset$, a prefix $\gamma(f)$, such
that in the graph induced by $\set{\gamma(f)\mid f\in \pset}$, the
prefixes appearing in each connected component either form a
canonical spider, a canonical cycle, or the connected component
contains exactly one prefix $\gamma(f)$, where $\gamma(f)=f$ for some $f\in
\pset$.
\end{theorem}

We simply apply Theorem~\ref{thm: CK-spider decomposition} to set $\pset_1\cup\pset_2$. Consider now some connected component $\cset$ in the graph induced by the prefixes. It is impossible that $\cset$ is a canonical cycle since a canonical cycle contains an odd number of paths where every pair of consecutive paths intersect. Therefore, each component $\cset$ either contains vertices of exactly one prefix $\gamma(p)$ and $p=\gamma(p)$ in this case, or it contains vertices of two prefixes $\gamma(p)$ and $\gamma(p')$ that form a canonical spider.

\section{The Integrality Gap of the LP}\label{sec:int-gap}
In this section we show a lower bound of $\Omega(\sqrt{m})$ on the integrality gap of the LP from Section~\ref{sec: almost feasible solutions}.
We then show how the algorithm described in Section~\ref{sec: final-algorithm} overcomes this gap.The construction of the gap example is somewhat similar to the construction used by \cite{BS} to show a lower bound of $\Omega(\sqrt{n})$ on the integrality gap of the configuration LP.

We describe a canonical instance together with an assignment of private items.
We start by describing a gadget $G$ that is later used in our construction. Gadget $G$ consists of
 $M$ light agents $L_1,\ldots,L_M$. For each light agent $L_j$, there is a distinct collection $S(L_j)$ of $M$ light items for which $L_j$ has utility $1$. Let $S=\cup_jS(L_j)$, note that $|S|=M^2$. Items in $S$ will not be assigned as private items to any agent.

Additionally, for each $j: 1\leq j\leq M$, agent $L_j$ has one heavy item $h(L_j)$, for which $L_j$ has utility $M$. This will also be $L_j$'s private item.
 The gadget also contains $M-1$ heavy agents $t_1,\ldots,t_{M-1}$. These heavy agents are not assigned private items and hence are terminals.
 Each terminal is a heavy agent that has utility $M$ for each one of the items $h(L_1),\ldots,h(L_M)$. Finally, we have a light agent $L^*$ that has utility $1$ for each item $h(L_1),\ldots,h(L_M)$.

We make $M$ copies of the gadget, $G_1,\ldots,G_M$. We denote the vertex $L^*$ in gadget $G_j$ by $L^*_j$.  We add a distinct heavy item $h(L^*_j)$ for each $L^*_j$. Item $h(L^*_j)$ is the private item for $L^*_j$, and gives utility $M$ to it. Finally, we have a heavy agent $t^*$ that has utility $M$ for each $h(L_j^*)$, $1\leq j\leq M$. This agent is also a terminal since it has no private item assigned.

Thus the set of terminal are all the heavy agents. The total number of items is $n=O(M^3)$ and the total number of agents
is $m=O(M^2)$.

We start by showing that in any integral solution, some agent receives a utility of at most $1$. That is, any integral flow from $S$ to the terminals will $1/M$-satisfy some light agent. This is because the terminal $t^*$ must receive one unit of flow from  $L^*_j$ for some $1\le j\le M$. Consider now the corresponding gadget $G_j$. We can assume w.l.o.g. that each light agent $L_i\in G_j$, $1\leq i\leq M$ receives $M$ flow units from its light items in $S(L_i)$. Each one of the $M-1$ terminals $t_1,\ldots,t_{M-1}$ has to receive one flow unit. This leaves only one flow unit to satisfy $L^*_j$, and so $L^*_j$ is assigned at most one item for which it has utility $1$.

We now argue that there is a fractional flow which $1$-satisfies all agents. Consider some gadget $G_i$. Each light agent $L_j\in G_i$ receives $M$ flow units from light agents in $S(L_j)$, and sends $1$ flow unit to its private item $h(L_j)$, which in turn sends $1/M$ flow units to each one of the agents $t_1,\ldots,t_{M-1},L^*_i$. Each one of the light agents $L^*_1,\ldots,L^*_M$ now receives $1$ flow unit can thus send $1/M$ flow to terminal $t^*$.

To be more precise, we have $y(L^*_j) = x(L^*_j) = 1/M$ for all $j=1,\ldots,M$ and for all $1\le j \le M$, for all $1 \le i \le M$, we have
$y(L^*_j,L^j_i) = 1/M$, where $L^j_i$ is the $i$th light agent in $G_j$. The flows are as described in the previous paragraph. One can check that this satisfies all the constraints of the LP described in Section \ref{sec: almost feasible solutions}.
Thus this is a feasible fractional solution in which each agent is $1$-satisfied, and the value of the solution for the \MMA problem is $M$. This completes the description of the gap example.

Before describing how our algorithm bypasses the integrality gap, we first show how in this example we can prove using the same LP that the integral optimum cannot be more than $1$. Firstly, note that removing any agent cannot {\em decrease} the integral optimum. Therefore, if we remove the set of light agents $\{L_1,\ldots,L_{M-1}\}$ from every $G_j$, the integral optimum should still be at least $M$. However, consider now the following assignment of private items -- for the remaining light agents we still have $P'(L^*_j) = h(L^*_j)$ and $P'(L^j_M) = h(L^j_M)$ for $1\le j\le M$, but now we assign a private item for every heavy agent $t_i$ in each gadget, $P'(t_i) = h(L_i)$, that is, the heavy item of agent $L_i$ (who is not present in this instance). The only terminal in this instance is the heavy agent $t^*$.

Note that the set $S$ of items which are not private items in each gadget $G_j$ is still $\bigcup_j S(L_j)$. However, the items in $S(L_j)$ for $1\le j\le (M-1)$ do not connect to any agent (since the light agents have been removed). Thus the flow to the terminal $t^*$ must come from the $M$ sets of items of the form $S(L^j_M)$.

We now argue that the LP of Section 4 is not feasible. In fact, even if $N_{L^*_j}$ for every gadget is reduced from $M$ to $2$, the LP is not feasible.
Since $t^*$ receives a flow of value $1$, it must receive a flow of at least $1/M$ from one of the $L^*_j$. Thus $x(L^*_j) \ge 1/M$.
This implies $L^*_j$ must receive $N_{L^*_j}\cdot x(L^*_j)$ units of flow from the light agents in the lower level. However, there is only one light agent,
namely $L^j_M$ in the lower level and from constraint \eqref{lambda-flow capacity constraints for one level} it can ``feed'' at most $x(L^*_j)$ units of
flow to $L^*_j$. Thus, if $N_{L^*_j} > 1$, the LP will be infeasible.\\

Now we show how after one iteration of the algorithm we get to the instance $(\I',P')$ described above.
After the rounding algorithm described in Section \ref{sec: almost feasible solutions}, we get set of paths $\pset_1,\pset_2$ which are as follows:
\begin{align*}
\pset_1 = (L^*_1 \to h(L^*_1) \to t^*) \cup \{(L^j_i \to h(L^j_i) \to t^j_i): \forall i=1\ldots M-1, 1\le j\le M\} \\
\pset_2 = \{(v\to L^j_i): v\in S(L^j_i), 1\le i \le M-1, 1\le j\le M\} \cup \{(L^1_i \to h(L^1_i) \to L^*_1): 1\le i\le M-1\}
\end{align*}
that is, $\pset_1$ is the set of paths from $L^*_1$ to $t^*$ and the paths from $L^j_i$ to $t^j_i$ in every gadget $G_j$; and $\pset_2$ is the
set of paths from $S(L_i)$ to $L_i$ in all gadgets and $L^1_i$ to $L^*_1$ for the gadget $G_1$. The path decomposition procedure of Section \ref{sec: final-algorithm}
returns one bad light agent ($L^*_1$) and the following sets of internally disjoint paths
\begin{align*}
\pset'_1 = \{(L^1_i \to h(L^1_i) \to t^1_i): \forall i=1\ldots M-1\} \\
\Q_2 = \{(v\to L^j_i): v\in S(L^j_i), 1\le i \le M-1,~1\le j\le M\} \cup \{(L^1_i \to h(L^1_i) \to L^*_1): 1\le i\le M-1\}
\end{align*}
Subsequently, the new set of terminals is $T_2 = \{t^*\}$ and the set of discarded light agents are $\lset_2 = \{L^j_i: 1\le i\le M-1, ~1\le j\le M\}$
and thus we get the instance $(\I',P')$. Hence, in the second iteration, the algorithm will return that the optimum $M$ is infeasible for this \MMA instance.

\end{document}